\documentclass{article}
\usepackage{fullpage}
\usepackage{mathtools}
\usepackage{amssymb}
\usepackage{amsfonts}
\usepackage{amsmath}
\usepackage{amsthm}
\usepackage{booktabs} %
\usepackage{wrapfig}
\usepackage{subcaption}
\usepackage{verbatim}
\usepackage{diffcoeff}
\usepackage{adjustbox}
\usepackage{url}
\usepackage{natbib}
\usepackage{authblk}
\usepackage{graphicx}
\usepackage[pdf]{graphviz}
\usepackage{multirow}
\usepackage{textcomp}

\usepackage{appendix}
\usepackage{color}
\definecolor{strcolor}{rgb}{0.6, 0.2, 0.6}
\definecolor{commentcolor}{rgb}{0.3125, 0.5, 0.3125}
\definecolor{keycol}{rgb}{0, 0, 1}

\newcommand{\rev}[1]{{#1}}

\usepackage{amssymb}
\usepackage{amsmath}
\usepackage{bbm}

\usepackage{listings}
\usepackage{url}

\newcommand {\bea}{\begin{eqnarray}}
\newcommand {\eea}{\end{eqnarray}}

\newtheorem{prop}{Proposition}

\def\blot{\quad \mbox{$\vcenter{ \vbox{ \hrule height.4pt
				\hbox{\vrule width.4pt height.9ex \kern.9ex \vrule width.4pt}
				\hrule height.4pt}}$}}

\usepackage{natbib}
\bibpunct[, ]{(}{)}{,}{a}{}{,}%

\gdef\AQ#1{}
\gdef\CQ#1{}

\usepackage{subcaption}
\usepackage{tikz}
\usepackage{siunitx}
\usepackage{float}
\usepackage{optidef}
\usepackage{booktabs}

\usetikzlibrary{decorations.text}
\newtheorem{remark}{Remark}
\newtheorem{assumption}{Assumption}

\newcommand{\arctt}{\boldsymbol{T}}
\newcommand{\pmbeta}{\boldsymbol{b}}

\newcommand{\rawnyc}{O}
\newcommand{\arcll}{\boldsymbol{L}}
\newcommand{\hbt}{{\hat{\boldsymbol{t}}}}
\newcommand{\hbx}{{\boldsymbol{x}}}

\title{Arc travel time and route choice model estimation subsumed}
\author[1]{Sobhan Mohammadpour}
\author[2]{Emma Frejinger}
\affil[1]{Massachusetts Institute of Technology}
\affil[2]{Université de Montréal}
\title{Arc Travel Time and Route Choice Model Estimation Subsumed}
\date{\today} 

\begin{document}

\maketitle
\begin{abstract}
    We address the problem of simultaneously estimating arc travel times in a network \emph{and} parameters of route choice models for strategic and tactical network planning purposes. Hitherto, these interdependent tasks have been approached separately \rev{in the literature on road traffic networks}. We illustrate that ignoring this interdependence can lead to \rev{erroneous} route choice model parameter estimates. We propose a method for maximum likelihood estimation to solve the simultaneous estimation problem that is applicable to any differentiable route choice model. Moreover, our approach allows to naturally mix observations at varying levels of granularity, including noisy or partial path data. Numerical results based on real taxi data from New York City show strong performance of our method, even in comparison to a benchmark method focused solely on arc travel time estimation.
\end{abstract}

\section{Introduction} \label{sec:intro}
An extensive body of literature in transportation science focuses on strategic and tactical planning to manage demand, for instance, through pricing, and to optimize network performance by addressing problems such as network design or facility location. For example, \cite{GilbertEtAl15} optimize prices in a network where users are assigned to paths according to a discrete choice model, assuming they minimize the cost, a function of the price they pay as well as other fixed attributes like travel time.  \cite{WeiEtAl22} examine a transit planning problem by considering competition from ride-hailing services and the effects of congested travel times. Additionally, \cite{OsorAtas21} propose a simulation-based approach for toll optimization, employing a simulator to accurately represent user behavior.

Common to these studies, and many others, is \emph{(i)} the need for \rev{travel time estimates} and \emph{(ii)} a model of user behaviour that distributes flow in the network, i.e., a route choice model. We provide an illustration in Figure~\ref{fig:processoverview} where $\arctt$ denotes a matrix whose elements represent travel time on an arcs in a network. Furthermore, $\mathcal{P}(o, d; \hat \arctt, \pmbeta)$ denotes a route choice model for any given origin $o$ and destination $d$, that depends on travel times (as well as other attributes) and a vector of parameters $\pmbeta$. In the right-hand side of the figure, we provide examples of strategic and tactical network planning problems where \rev{estimated} travel times can be considered fixed and given $\hat \arctt$, as in \cite{GilbertEtAl15}. Alternatively, travel time can vary with traffic flow within an assignment/equilibrium model, or a simulator as exemplified in \cite{OsorAtas21} and \cite{WeiEtAl22}. 

\rev{Strategic and tactical planning problems are challenging to solve namely because the demand distributions -- here traffic flow distributions -- depend on supply decisions \citep{FrejHewi25}.}
\rev{In our context, such} distributions are given by a route choice model whose parameter values $\hat \pmbeta$ have been estimated based on data. \rev{These are typically structured models, such as discrete choice models, to allow for interpretable out-of-distribution generalization.} The two estimation problems -- travel time and route choice model estimation -- are interdependent. As illustrated in Figure~\ref{fig:processoverview}, the route choice model is based on network attributes and travel time (or speed) is known to be an important explanatory factor. \rev{In turn, if paths are not perfectly observed} the travel time estimation task requires a route choice assumption \rev{to match trip observations to the network representation}.

This leads us to the focus of this paper which is illustrated with a dotted frame in the figure. Consider a network of nodes and arcs representing a given congested geographical area. \rev{We aim to simultaneously estimate arc travel times and route choice model parameters for large-scale road traffic networks using data at different levels of granularity. At the lowest level of granularity, we assume that trip observations contain origin and destination information along with travel time (i.e., the path is latent). Higher levels of granularity assume, in addition, GPS sequence information that captures fully or partially the chosen paths. The resulting travel time estimates and route choice model are intended to support strategic and tactical network planning. It is therefore crucial to ensure that arc travel times are additive and that their sums lead to accurate path travel time estimates.} 

\begin{figure}
    \centering
    \includegraphics[width=0.8\textwidth]{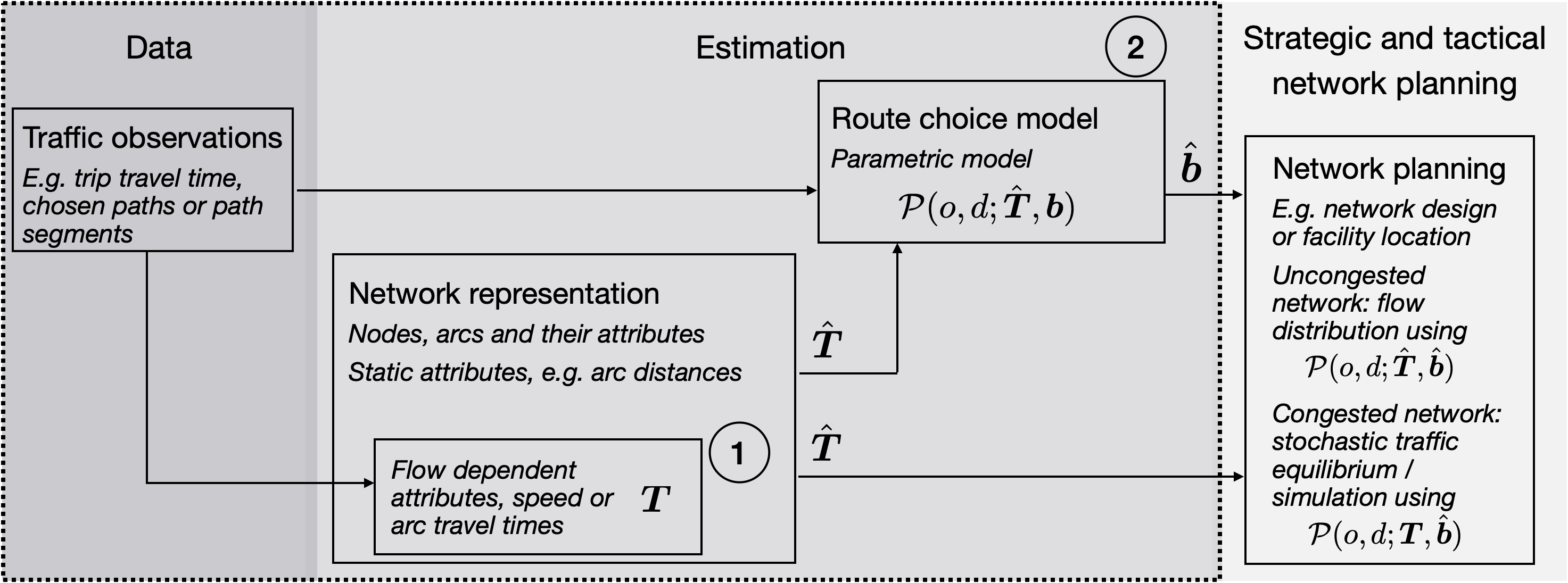}
    \caption{Overview of data, travel time and route choice model estimation and their link to network planning problems}
    \label{fig:processoverview}
\end{figure}

\rev{The interdependence between travel time and route choice model estimation has been recognized in the literature on passenger flow assignment in public transport networks \citep[e.g.,][]{SunEtAl15, MoEtAl23, ChenEtAl25}. It arises because Automatic Fare Collection (AFC) data includes only tap-in at tap-out times and the origin and destination stations, respectively. The primary objective is then to infer the latent path choice for each trip observation. Our work is, to the best of our knowledge, the first to tackle this simultaneous estimation problem for road traffic network with a double primary objective: estimate arc travel times in the network and parameters of a stochastic route choice model. We outline the related difficulties by following closely the five key challenges described in \cite{ChenEtAl25}.}

\rev{The \emph{first} challenge pertains to \emph{unobserved path choices}. While chosen paths are latent in public transport AFC data, this is not necessarily the case in road traffic networks thanks to GPS tracking. It is, however, common to have imperfect path information (e.g., imprecise or low-frequency GPS sequences) or to lack detailed GPS data for privacy reasons. An approach that can handle data at different levels of granularity is therefore useful. The \emph{second} challenge is that \emph{network travel times are unobserved and dynamic}. Even large GPS probe vehicle data sets, and networks with sensors, cannot cover an entire large-scale network over time. The \emph{third} challenge concerns \emph{heterogeneous path choice preferences}. These vary across the population and may even change dynamically for the same individual (e.g., depending on trip purpose). A \emph{fourth} challenge arises from the \emph{aleatoric and epistemic uncertainty} associated with demand and supply in transport networks. This setting calls for probabilistic approaches. Finally, the \emph{fifth} challenge is the \emph{high dimensionality} of the estimation problems.}

\rev{\cite{ChenEtAl25} propose a Bayesian hierarchical model to address these challenges for passenger trip assignment in metro networks. Unlike other works that focus on static settings, they estimate dynamic network costs and a multinomial logit path choice model with spatiotemporally varying coefficients. Their estimation problem is high-dimensional even though they decompose observed trip travel times into only four components (access, in-vehicle, transfer, and egress time). Road traffic networks do not allow for such a decomposition, which makes the high dimensionality of the problem a major challenge.}

\rev{There is extensive literature on road traffic networks that treats the two estimation problems separately. Existing methods can then be used in what we refer to as a} \emph{two-step approach} (illustrated with the numbers in \rev{Figure~\ref{fig:processoverview}}). It consists of first estimating travel times, and then, treating those estimates as fixed, \rev{estimating the route choice model parameters in a second step.} 

\rev{As in the passenger flow assignment problem, if paths are unobserved, then the travel time estimation problem requires a path choice assumption in the first step.} Closest to our work is the estimation approach in \cite{bertsimas2019travel}. Like us, they aim to estimate travel times on all arcs in a network \rev{for tactical and strategic planning, using only origin, destination, and trip travel time information}. \rev{They assume shortest path choice (i.e., a deterministic choice model) and report results for the New York City taxi data. Their approach is designed for data where paths are unobserved, and it scales  to large networks. Hence, it addresses the first and fifth challenges, but only partially the second challenge (due to static network assumption). It does not address the third challenge, since shortest-path choice is assumed, and consequently it does not fully address the fourth (partially ignoring uncertainty).}

\rev{As we further detail in Section~\ref{section:litreview}, when path observations are available, a wide range of route choice models can be used in the second step. The options are more restricted when paths are only partially observed or when only arc flows are available. To the best of our knowledge, no prior work has estimated route choice models using only origin, destination, and trip duration information.}

The paper offers the following contributions:
\begin{itemize}
\item \rev{We introduce the problem of simultaneously estimating arc travel times and route choice model parameters for road traffic networks. Whereas this problem has been studied for public transport networks, related approaches \citep[e.g.,][]{SunEtAl15, MoEtAl23, ChenEtAl25} do not directly apply to large-scale road networks. Indeed, in this case travel time cannot be naturally decomposed into a small number of components, and it is not possible to leverage schedule information.}
\item We propose a methodology to estimate arc travel times and parameters of differentiable route choice models simultaneously. \rev{We do not restrict the approach to a specific type of route choice model, like multinomial logit}. We show that under weak assumptions, we can conveniently formulate a maximum likelihood estimator for the \rev{simultaneous} estimation problem. Furthermore, we show that the solution to the model in \citet{bertsimas2019travel}, under certain assumptions, results in a maximum likelihood estimate. \rev{Similar to \citet{bertsimas2019travel}, we assume a static network over a given time period (e.g., morning peak hours). Unlike them, we allow for heterogeneous path choice preferences.}  
\item An attractive property of our proposed joint likelihood is that it can naturally mix observations of different levels of granularity. This can include incomplete or noisy data, and generalizes the methodology of \cite{mai2023estimation} \rev{for estimation of route choice models} with partial path observations.
\item Numerical results based on New York City taxi data show that our method is fast and performant. First, focusing on the travel time estimation task only, measured by mean-squared log error, we reach a performance and computing times comparable to those of \citet{bertsimas2019travel} even though we solve the \rev{simultaneous} estimation problem. Focusing on both tasks, results on synthetic data illustrate that our travel time estimates better reproduce ground truth values compared to those obtained using the two-step approach.
\end{itemize}

\rev{With respect to the aforementioned challenges, we allow for data at different levels of granularity including unobserved path choices hence addressing the first challenge. We partially address the second challenge by estimating network travel times. However, we assume a static network representation for a given time period. This is aligned with \cite{bertsimas2019travel} but is more restrictive than \cite{ChenEtAl25}. Our approach is applicable to any route choice model that is differentiable. Allowing stochastic choice models and opening up for modeling heterogeneous path choice preferences (third challenge) is a key contribution in comparison to \cite{bertsimas2019travel}. Similar to works on passenger flow assignment \citep[][]{ChenEtAl25,MoEtAl23,SunEtAl15}, we use a multinomial logit model in our results \citep[in our case a recursive logit model to allow for unrestricted choice sets][]{fosgerau2013link}. Our approach is probabilistic and scales to large networks, hence, to our knowledge, we are the first to address the fourth and fifth challenges for road traffic networks.} 

There are multiple implications of this work. \cite{ramos2020route} show empirically that \rev{inaccuracies in} travel time estimates $\hat \arctt$ used as input for route choice modeling \rev{may be absorbed by} route choice model parameters $\hat \pmbeta$. Since estimated route choice models are employed to address planning problems, \rev{erroneous} parameter estimates can adversely affect the resulting solutions. For instance, overestimating traffic flow in a specific area of the network can result in unnecessary infrastructure investments. Our proposed simultaneous estimation approach \rev{addresses the inaccuracies arising from} the interdependence between the two estimation problems. Moreover, since our approach allows mixing observations at different levels of granularity, it reduces the data requirements and mitigates privacy concerns related to tracking the position of individual drivers. To the best of our knowledge, this is the first work to estimate a route choice model for a road traffic network using only travel time observations.

The paper is structured as follows. In Section~\ref{section:litreview}, we describe related work on travel time estimation with a particular focus on \citet{bertsimas2019travel}, along with a background on route choice modeling.
In Section~\ref{section:motivation}, we provide an illustrative example, and in Section~\ref{section:methodology}, we introduce our method. We report numerical experiments in Section~\ref{section:experiments}, and finally, we provide concluding remarks in Section~\ref{sec:Conclusion}.

\section{Related Work}\label{section:litreview}
\rev{Given our focus on the simultaneous estimation problem, our work lies at the intersection of travel time estimation and route choice model estimation. As highlighted in the introductory section, at this intersection we are only aware of studies on passenger flow assignment \citep[e.g.,][]{ChenEtAl25,MoEtAl23,SunEtAl15}. Their primary objective is to assign trip observations to paths in the public transport network, but they also estimate travel times (or generalized costs) in the process. As highlighted in the introductory section, whereas we share the problem definition, the structure of public transport networks differs from that of road traffic networks. Furthermore, GPS data may provide more granular trip observations compared to AFC data. Thus, the approaches on passenger flow assignment do not directly apply to our setting.}

\rev{Therefore, this section reviews the extensive literature that treats the two estimation problems in isolation. We begin with a brief overview of travel time estimation to motivate why we consider \cite{bertsimas2019travel} the work most closely related to ours. We benchmark our travel time estimation results against those obtained using their approach, and for this reason, we describe their methodology in Section~\ref{sec:litRevBertsimas}. Our approach applies to differentiable route choice models. In Section~\ref{section:pathchoicemodel}, we review route choice models and highlight the complementarity of our work.}

\rev{Travel time estimation is a fundamental problem for advanced traveler information systems and intelligent transport systems, and the related literature is rich \citep{shaygan2022traffic}. The problem consists of predicting the travel time for \emph{a given trip}. There are different variants: predict travel time for a given path or for a specific road segment \citep{YanEtAl24}, or predict travel time for a pair of locations without knowledge of the path. Since the predictions are trip specific, the prediction algorithms may leverage contextual features such as weather and departure time. If departure time is known, then this problem corresponds to predicting the estimated time of arrival (ETA).} 

\rev{Closest to our work are approaches that predict travel time for given road segments and then aggregate travel time predictions to construct a path-level prediction. In our case, we focus on the highest level of geographical granularity, where a segment corresponds to an arc in the network representation, and predictions cover \emph{all} arcs. This differs from \emph{trip} travel time predictions that do not necessarily require this level of geographical granularity and scope. Instead, it is possible to directly predict travel for the entire path, for instance, using deep learning algorithms \citep[e.g.,][]{JindalEtAl17,YuanEtAl20}, or for segment or super-segments \citep[graph neural networks are deployed for ETA prediction in Google Maps,][]{Derrow-PinionEtAl21,XueEtAl25}. \cite{YanEtAl24} provide a theoretical analysis of different methods and conclude that segment-level granularity is often of first-order importance for predictive accuracy.}

\rev{Deep learning is also used for probabilistic regression, which provides uncertainty estimates in addition to point predictions. \cite{xu2025link} provide a recent literature review and propose an approach dubbed ProbTTE. Using data in the form of GPS sequences, ProbTTE can characterize both inter-trip and intra-trip correlations by modeling trip-level arc travel times with a Gaussian hierarchical model. Unlike previous works, and similar to us, they sum arc travel times to estimate path travel times, i.e., their trip representation is a path-based sum of arc representations. Other works propose approaches to capture spatial and/or temporal correlations but do not impose linear arc travel times \citep[for two examples, see][]{SunKim21,ZhouEtAl23}. Moreover, such approaches often require sequences of observed locations and times for a given trip, and are therefore limited in their ability to handle data where only total trip travel times are observed.}

\rev{Compared to the aforementioned literature, which focuses on predicting travel time for a given trip, relatively few works address the same estimation task as ours: namely, estimating arc travel times across an entire network without assuming access to path observations, and while ensuring that arc travel time estimates can be summed to recover path travel times. \cite{ZhanEtAl13}, \cite{bertsimas2019travel}, and \cite{GhanKouv22} study this problem using the New York City taxi dataset. These approaches are all based on a deterministic shortest path assumption to map trip travel time observations to the network representation. \cite{bertsimas2019travel} show that their approach outperforms the one proposed in \cite{ZhanEtAl13}. More recently, \cite{GhanKouv22} assume that arc travel times follow a Gaussian distribution. They estimate the means and variances of these distributions using maximum likelihood while also modeling spatial correlation. However, there are scalability concerns, as they report results only for a small part of the Manhattan network, and they do not benchmark their approach against \cite{bertsimas2019travel}.}

\rev{In sum, to our knowledge, our study is the first to address the simultaneous estimation problem for large-scale road traffic networks. Most of the literature on travel time estimation focuses on predicting travel time, or ETA, for a given trip, which is a short-term prediction problem. In contrast, relatively few works tackle our formulation of the travel time estimation problem, and they rely on a deterministic path choice assumption. We aim to relax this assumption by jointly estimating arc travel times and a probabilistic path choice model. Given our focus on large-scale networks, we benchmark the performance of our travel time estimates and computing times against \cite{bertsimas2019travel}. We therefore describe their method in more detail next.}

\subsection{A Method for Arc Travel Time Estimation}\label{sec:litRevBertsimas}

\cite{bertsimas2019travel} propose a methodology to find a point estimate of all arc travel times in a network based on observations of trip travel times between various origin-destination (OD) pairs. While detailed path observations are not required, the methodology can incorporate such data if available. Assuming that a path travel time estimate is the travel time of the shortest path, they estimate the parameters by minimizing divergence between estimated and observed travel times, as we describe in more detail next.

We represent the network with a graph $G=(N, A)$ where $N$ and $A$ are the set of nodes and arcs, respectively. We denote the set of all paths from origin $o\in N$ to destination $d\in N$ as $R(o,d)$. The list of all observations is denoted by $\rawnyc$ and contains tuples $(o, d, t) \in N^2\times\mathbb{R}_+$ where $t$ is the travel time from  $o$ to $d$. We denote the set of all OD pairs in $\rawnyc$ as $W$. 

\citet{bertsimas2019travel} estimate the travel time matrix $\arctt_{ij}$ for all arcs $(i,j) \in A$ using $\rawnyc$ that minimize the divergence $l(\hat{t}, t)$ between the observation $t$ and the estimate $\hat{t}$ under the assumption that $\hat{t}$ is the length of the shortest path. The length of the shortest path is obtained by having a binary variable $z_r$ for every path $r$ and linear (including big $M$) constraints. The problem is formulated as a mixed-integer program and we refer to it as BDJM \emph{model} (from the first letter of each contributing author's name):
\begin{mini*}
    {\arctt,\hbt,\boldsymbol{z}}
    {\sum_{(o,d,t)\in\rawnyc}l(\hbt_{od}, t)}
    {}
    {}
    \addConstraint{\hbt_{od}\leq}{\sum_{(i,j)\in r}\arctt_{ij}}{\quad\forall(o,d)\in W, r\in R(o,d)}
    \addConstraint{\sum_{(i,j)\in r}\arctt_{ij}\leq}{\hbt_{od} +  M(1 - \boldsymbol{z}_{r})}{\quad\forall(o,d)\in W, r\in R(o,d)}
    \addConstraint{\sum_{r\in R(o,d)} \boldsymbol{z}_r = }{1}{\quad\forall (o,d) \in W}
    \addConstraint{\boldsymbol{z}_r \in }{\{0,1\}}{\quad\forall (o,d)\in W, r\in R(o,d)}
    \addConstraint{\arctt_{ij} \geq }{0}{\quad(i,j)\in A}
    \addConstraint{\hbt_{od} \geq }{0}{\quad(o,d)\in W}.
\end{mini*}

Based on the assumption that the quality of travel time estimations is perceived on a multiplicative rather than an additive scale (meaning that observing a travel time of 30 seconds when the estimated travel time is 1 minute is as bad as observing a travel time of 30 minutes when the estimated travel time is 1 hour), they propose using the mean-squared log error (MSLE), 
\begin{equation}
    \mathrm{MSLE}(\hat{t}, t) = (\ln \hat{t} - \ln {t})^2 = \left[\ln(\hat{t}/{t}) \right]^2,
\end{equation}
as the divergence function. One advantage of using the MSLE is that it allows for aggregation of all the travel time observations for the same OD by taking their geometrical mean. This is possible because the MSLE can be decomposed into the MSLE between the prediction and the geometrical mean $\bar{t}$, plus the MSLE between the observations and the geometrical mean:
\begin{equation}
    \sum_{i=1}^n \mathrm{MSLE}(\hat{t}, \boldsymbol{t}_i) = n\mathrm{MSLE}(\hat{t}, \bar{t}) +\sum_{i=1}^n \mathrm{MSLE}(\bar{t}, \boldsymbol{t}_i).
\end{equation}
The second part of the term is independent of the optimization variables and can hence be omitted.

To this end, we introduce $n_{od}$, the number of observations for the OD $od$ and $t_{od}$, the geometrical mean of the observations from that OD.
\cite{bertsimas2019travel} introduce the convex surrogate 
\begin{equation}\label{eq:Bertsimas:surrogate}
    \max\left\{\frac{\hat{t}}{\bar{t}}, \frac{\bar{t}}{\hat{t}}\right\},
\end{equation} 
for $\mathrm{MSLE}(\hat{t}, \bar{t}),$ which they model using a second-order cone and linear constraints. A new variable 
$\boldsymbol{x}_{od}$, greater than both $\boldsymbol{\hat{t}}_{od}/\boldsymbol{\bar{t}}_{od}$ and $\boldsymbol{\bar{t}}_{od}/\boldsymbol{\hat{t}}_{od}$, is minimized. The constraint $\boldsymbol{x}_{od}\geq\boldsymbol{\bar{t}}_{od}/\boldsymbol{\hat{t}}_{od}$ \rev{is equivalent to} the second-order cone constraint
$$\boldsymbol{x}_{od} + \hat{\boldsymbol{t}}_{od} \geq \|(\boldsymbol{x}_{od} - \boldsymbol{\hat{t}}_{od}, 2\sqrt{\bar{\boldsymbol{t}}_{od}})\|$$ if all variables are nonnegative, and  $\boldsymbol{x}_{od}\geq\hat{\boldsymbol{x}}_{od}/\bar{\boldsymbol{x}}_{od}$ is a linear constraint.

The resulting estimation problem is undetermined by nature as the travel time of arcs not included in any shortest path is only bounded from below. Hence, regularization may play an important role. \citet{bertsimas2019travel} define the regularization on the sequence of arcs $(i,j)$ and $(j, k)$,
\begin{align}\label{eq:Bertsimasreg}
    \left|\frac{\arctt_{ij}}{\arcll_{ij}} - \frac{\arctt_{jk}}{\arcll_{jk}}\right|\frac{2}{\arcll_{ij} + \arcll_{jk}},
\end{align}
where $\arctt_{ij}$ and $\arcll_{ij}$ are, respectively, the travel time and the length of arc $(i,j)\in A$. The regularization for the problem is the sum of the regularization for all consecutive arcs. The intuition is that the speeds on any two adjacent arcs should not be drastically different.

Note that there is an exponential number of binary variables in BDJM model. The authors, therefore, propose an algorithm based on iterative path generation, which results in solving the second-order conic program (\ref{prog:BDJM}), assuming a fixed shortest path $r^*_{od}$ and a limited set of paths $K(o,d)\subseteq R(o,d)$: 
\begin{mini!} 
    {\arctt,\hbt,\hbx}
    {\sum_{(o,d)\in W} \boldsymbol{n}_{od} \hbx_{od} + \lambda \sum_{(i,j)\in E}\sum_{(j,k)\in E}\left|\frac{\arctt_{ij}}{\arcll_{ij}} - \frac{\arctt_{jk}}{\arcll_{jk}} \right| \frac{\rev{2}}{\arcll_{ij} + \arcll_{jk}}}
    {\label{prog:BDJM}}
    {}
    \addConstraint{\hbt_{od}\leq}{\sum_{(i,j)\in r}\arctt_{ij}}
    {\quad\forall(o,d)\in W, r\in K(o,d)}
    \addConstraint{\hbt_{od}=}{\sum_{(i,j)\in r^*_{od}}\arctt_{ij}}{\quad\forall(o,d)\in W}
    \addConstraint{\hbx_{od}\geq}{\hbt_{od} / \boldsymbol{t}_{od}}{\quad\forall (o,d) \in W}
    \addConstraint{\boldsymbol{t}_{od}\leq}{\hbt_{od}\hbx_{od}}{\quad\forall (o,d) \in W\label{prog:BDJM:xt}}
    \addConstraint{\arctt_{ij} \geq }{0}{\quad\forall (i,j)\in A }
    \addConstraint{\hbt_{od} \geq }{0}{\quad\forall (o,d) \in W}
    \addConstraint{\hbx_{od}\geq }{0}{\quad\forall (o,d) \in W},
\end{mini!}
where $\boldsymbol{x}_{od}$ becomes equal to the surrogate at the optimum.

At each step, the relaxed program (\ref{prog:BDJM}) is solved with a fixed shortest path $\boldsymbol{r}^*_{od}$, then assuming $\arctt$ is fixed, the shortest path is updated and added to $K(o,d)$. We call this the BDJM \emph{method}.

\citet{bertsimas2019travel} evaluate their method on New York City's yellow cab dataset \citep{taxi2019} limited to Manhattan and obtain high-quality results. Solving the problem is challenging because the network and the corresponding \rev{dataset} are large. We use their method and the yellow cab dataset \rev{to benchmark travel time estimation.}

\subsection{Route Choice Models \label{section:pathchoicemodel}}\label{FosgProb}

A route choice model, $\mathcal{P}(o, d; \arctt, \pmbeta)$, defines a probability distribution over all paths $r\in R(o,d)$ that is parameterized by $\pmbeta$. It is based on a network representation, and in most studies, related attributes, including travel time, are assumed to be fixed and known (i.e., a static and deterministic setting). Exceptions include studies on stochastic travel times \citep[e.g.,][]{GaoFrejBenA10,ding2018latent,mai2021routing} \rev{where travel time distributions are assumed to be known}. Key route choice modeling challenges pertain to the large size (infinite if we consider paths with cycles) of $R(o,d)$, and the similarity of paths, for instance, due to physical overlap. \cite{oyama2017markovian} reports roughly $10^{18}$ paths in a ten by ten grid which clearly shows that full path enumeration is not possible even for small networks. 

We separate existing route choice models into three broad classes. First, models that are based on restricted sets of paths. The sets can be restricted in different ways through what is often referred to as choice set generation. The models in this class can be path-based or arc-based \citep[e.g.,][]{dial1971probabilistic}. Models in the second class consider path choice as a sequence of arc-choices following a Markov decision process \citep{fosgerau2013link}. These models are based on the full network and hence do not impose any restriction on the path set. \cite{fosgerau2022perturbed} introduce perturbed utility route choice models, a model and solution method that we consider belonging to a third class. They show that route choice model parameters can be estimated using observations of flow vectors directly.

The classical way of estimating models for the first two classes is maximum likelihood estimation over path observations. The perturbed utility route choice model can advantageously be estimated by linear regression \citep{fosgerau2022perturbed}. However, we note it could also be estimated by maximum likelihood as the model can produce path choice probabilities. In terms of observations for maximum likelihood estimation, paths can be fully, or partially observed \citep[e.g.,][]{BierFrej08, mai2023estimation}. Partial observations may arise for various reasons, for instance, due to the measurement method (e.g., cameras placed in various intersections in the cities), measurement errors (GPS tracking is not always accurate, especially in dense urban zones), or low frequency of measurement (it might be expensive to measure the location of the car at high frequency). In such settings, \cite{mai2023estimation} show that a model can still be estimated using gradient estimation and sampling.

In this paper, we limit our methodology to route choice models whose parameters are estimated by maximum likelihood. \rev{While this encompasses many route choice models, it excludes certain ones, for example generative adversarial networks \citep[such as][]{choi2021trajgail}.} Moreover, in Section~\ref{section:solution} (solution approach) and onward, we restrict ourselves to differentiable route choice models. The models in the second class are often differentiable, and so is the model of the third class \citep{fosgerau2022perturbed}. The latter is motivated by the fact that convex conic problems (including convex quadratic programs) are differentiable with respect to their decision variables. Hence the resulting Markovian policy is differentiable \citep{amos2017optnet,agrawal2019differentiable}. Models in the first class are differentiable if the choice set generation algorithm does not depend on travel time. We note that many choice set generation algorithms are based on some variant of a repeated shortest path search and hence can use distance instead of travel time. This would result in a differentiable model.

\rev{Depending on data availability, heterogeneous preferences can be modeled in different ways (referred to as the third challenge in Section~\ref{sec:intro}). While any stochastic model -- even the simple multinomial logit -- relaxes the deterministic shortest path assumption, it does not model heterogeneous preferences in the population if the parameter values are constant across all individuals. In contrast, mixed logit models are flexible and can account for heterogeneous preferences and spatial correlation \citep[][]{Train2002,mai2018decomposition,ZimmEtAl18}. They are, however, computationally more costly to estimate because they require simulation to evaluate the likelihood function.} 

\rev{While our approach is not limited to the multinomial logit model,} we use the recursive logit model in \cite{fosgerau2013link} for the experimental section of this paper. There are \rev{three} main reasons for this choice. First, it is differentiable and fast to evaluate. \rev{Second, the variability in the New York data that we use for benchmark reasons is limited as it only covers taxi drivers.} \rev{Third}, this model, or its variants are used in many studies \citep[e.g.,][]{zhang2021type, cortes2023recursive,oyama2023capturing,knies2022recursive,gao2021estimation,koch2020review,iizuka2020cost,mai2015nested,mai2018decomposition,mai2023estimation,de2019modelling}. 

We describe our methodology in Section~\ref{section:methodology}, after the following motivating example.

\section{Motivating Example: Two Arcs, One Hard Network}\label{section:motivation}

\begin{figure}
    \centering
        \centering
    \begin{subfigure}[t]{0.45\linewidth}
    \centering
\begin{tikzpicture}
  [scale=.8,auto=left,every node/.style={circle,fill=blue!20}]
  \node (o) at (1,10) {o};
  \node (d) at (9,10) {d};
  
  \draw [->] (o) to [bend left] node[midway, fill=none, draw=none] {$t_1$} (d) ;
  \draw [->] (o) to [bend right] node[midway, below,fill=none, draw=none] {$t_2$} (d) ;
\end{tikzpicture}
\caption[A two-arc network.]{A two-arc network with one OD pair. Arc travel times, $t_1$ and $t_2$, are given next to each arc.}
\label{fig:counter-graph}
\end{subfigure}
\hfill
    \begin{subfigure}[t]{0.45\linewidth}
    \centering
    \includegraphics[width=\linewidth]{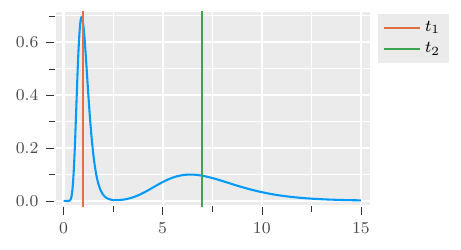}
    \caption{Probability density function of the noisy observations. Vertical lines are the ground truth travel times.\label{fig:obs_pdf2stone}}
    \end{subfigure}
\caption{Network and noise model visualization.}
\end{figure}

The motivation behind our work is the interdependence between the problems of estimating arc travel times and route choice model parameters: Arc travel times are needed to estimate the route choice model, and a route choice model is needed to estimate the arc travel times, especially when full path observations are unavailable. With a small example network (see  Figure~\ref{fig:counter-graph}), we illustrate that the two-step approach using data without path observations (i.e., only travel time is observed) cannot find the ground truth parameter values.

In this small network, the travelers at $o$ can take either the top or bottom path with respective travel times $t_1$ and $t_2$. The two paths are assumed to be independent \rev{(i.e., their utility distributions are independent and the travel times are independent as the paths do not physically overlap)} except that we enforce $t_1\leq t_2$. We assume that the probability of taking path 1 is $\exp(\beta t_1) / (\exp(\beta t_1) + \exp(\beta t_2))$ under a logit model with parameter $\beta$. The true parameters are $\beta=-0.2$, $t_1=1$, and $t_2=7$. Under these assumptions, the probability of taking path 1 is approximately $0.77$. 

We assume access to noisy travel time observations, and we do not observe the path choice. We posit a multiplicative log-normal ($\mu=0$, $\sigma=0.3$) noise model that is visualized in Figure~\ref{fig:obs_pdf2stone} and we sample $10,000$ travel time observations from this distribution. For this simple example, we compute maximum likelihood estimates by doing a grid search with a granularity of $0.1$ over the domain $(\beta,t_1,t_2)\in [-10, 0]\times(0, 10]\times(0, 10]$. We illustrate different maximum likelihood estimates by freezing certain parameters in the grid search.

Consider the two-step approach. To estimate travel time parameters we need a route choice assumption. We assume that it is close to a shortest path by fixing $\beta$ to $-10$, and we keep this value fixed when performing a grid search for $t_1$ and $t_2$. The resulting maximum likelihood estimate $(\hat{t}_1,\hat{t}_2)$ is then approximately $(1.5, 1.9)$. Keeping travel times fixed to those values in the second step yields a maximum likelihood estimate of $\hat{\beta} \approx -1.9$. All three parameter estimates are far from their ground truth values. 

If we instead consider the simultaneous estimation problem (i.e., we perform a grid search on the three parameters jointly), the maximum likelihood is then $(\hat{\beta}, \hat{t}_1,\hat{t}_2)\approx(-0.2, 1.0,7.0)$. It corresponds to the ground truth values. In this case, the high precision is due to a large \rev{dataset} for a very small network. Nevertheless, the example clearly illustrates that ignoring the interdependence between these estimation problems can lead to \rev{inaccurate} estimates.

Note that if we were iterating over the two steps in this toy example, it can be viewed as using coordinate descent \citep{Wright15}. However, the BDJM method cannot accommodate such an iterative scheme as it is assumed that the probabilities do not change when estimating the travel time. Indeed, it could cause the model to terminate at nonoptimal solutions or to diverge. Lastly, we note that existing methods cannot exploit travel time information to improve path choice models yet, as shown, travel time carries information about the choice made.

\section{Methodology}\label{section:methodology}

In this section, we introduce a mixture that models trips in the network. We analyze the properties of the mixture and propose a solution approach. Whereas our methodology differs from both \citet{bertsimas2019travel} and the route choice model estimation literature, we reuse notation from Section~\ref{section:litreview} unless stated otherwise.

\subsection{A Mixture to Model Observations at Different Granularities}

We represent every \rev{observed} trip in \rev{the network represented by} $G$ as a four-tuple 
\begin{equation}
    (o,d,t,r) \in N^2\times \mathbb{R}_+\times R(o,d),
\end{equation}
where, like before, $o$ is the origin, $d$ is the destination, $t$ is a travel time observation, and $r$ is the \rev{chosen} path. 
We assume that they follow a mixture of distributions:
\begin{subequations}\label{eqs:mixture}
\begin{align}
    &o \sim \mathcal{O}, \\
    &d \sim \mathcal{D}(o), \\
    &r \sim \mathcal{P}(o, d; \arctt, \pmbeta), \\
    &\rev{t \sim \mathcal{Y}(r; \arctt)}.
\end{align}
\end{subequations}
Origin $o$ follows a distribution $\mathcal{O}$ over the set of all nodes, and $d$ follows a distribution $\mathcal{D}(o)$ over the set of all nodes except $o$. Path $r$ follows a distribution over the set of all paths $R(o,d)$ described by a route choice model $\mathcal{P}(o, d; \arctt, \pmbeta)$ that is parameterized by $\pmbeta$. To simplify the notation we omit explicitly defining other features as they are assumed to be fixed and given. \rev{The distribution of the observed travel time $t$ is a function of $r$ and $\arctt$. We require $t$ to have a} positive support, preferably $\mathbb{R}_+,$ to ensure that all observations have a non-zero likelihood. \rev{The travel time distribution} models the difference in travel time due to factors such as fluctuating arc travel time, non-optimal habits, and speeding. 

Mixture~\eqref{eqs:mixture} builds on the assumption that observations are independently and identically distributed (i.i.d.). \rev{As highlighted in Section~\ref{section:litreview}, with the exception of a few works \citep[e.g.,][]{xu2025link}, this is a standard assumption. W}e note that it typically does not hold, since a second trip among two consecutive trips often originates from the first trip's destination. The assumption could be relaxed, for example, if we \rev{had} driver-specific information. In this case, we can make $\pmbeta$ a function of those features. \rev{We note, however, that empirical evidence suggests that such intertrip correlation is much weaker than intratrip correlation for observations occurring at similar times of travel \citep[see Figure~5 in][]{WoodEtAl17}. In this work, we assume that $\mathcal{Y}$ is a univariate distribution parameterized by $h = \sum_{(i, j) \in r} \arctt_{ij}.$ This structural form is appealing because of its simplicity but does not capture any correlation. We note that, depending on the granularity of the dataset, descriptive analyses (such as correlation analysis) may help verify the adequacy of the statistical assumptions. In our experimental results, we analyze the benefits of our approach by comparing the quality of the model output against ground-truth values (synthetic data) and the output of a benchmark method.}

Given the mixture~(\ref{eqs:mixture}), the likelihood of an observation $(o, d, t, r)$ is
\begin{equation} \label{eq:genLL}
f(o,d,r,t) = \mathbb{P}(o)\mathbb{P}(d|o)\mathbb{P}(r|o,d){\color{blue}f(t;r)},
\end{equation}
where $\color{blue}f(t;r)$ is likelihood function of $t$ parameterized by $\rev{r}$. In the following we provide several examples illustrating the flexibility of this mixture: It allows to model observations at different levels of granularity, e.g., paths or travel times are not observed, or paths are only partially observed. It also allows to combine different likelihoods, that is, combine in a singe \rev{dataset}, observations of different types.

If we do not observe any of the elements, we marginalize them. For instance, for an observation $(o, d, t)$ where the path information is missing, the likelihood is
\begin{align}
f(o,d,t) &= \sum_{r \in R(o, d)} \mathbb{P}(o)\mathbb{P}(d|o)\mathbb{P}(r|o,d)f(t;r) \\
&= \mathbb{P}(o)\mathbb{P}(d|o)\mathbb{E}_{r\sim\mathcal{P}}\left[f(t;r)\right] \label{nycll},
\end{align}
where $\mathcal{P}$ is shorthand for $\mathcal{P}(o, d; \arctt, \pmbeta)$. The corresponding log-likelihood is
\begin{equation} \label{eq:ll}
    \ln f(o, d, t) =   \ln\mathbb{P}(o) + \ln\mathbb{P}(d|o) + \ln\mathbb{E}_{r\sim\mathcal{P}}\left[f(t;r)\right].
\end{equation}

Similarly, if we observe the path but not its travel time, the likelihood is
\begin{align} \label{eq:marginalized_t}
    f(o,d,r) 
    &= \int_0^\infty f(o,d,r,t) \mathrm{d}t \\
    &= \int_0^\infty \mathbb{P}(o)\mathbb{P}(d|o)\mathbb{P}(r|o,d)f(t;r)\mathrm{d}t \\
    &= \mathbb{P}(o)\mathbb{P}(d|o)\mathbb{P}(r|o,d)\int_0^\infty f(t;r)\mathrm{d}t \label{eqs:normal_prob:int}\\
    &=  \mathbb{P}(o)\mathbb{P}(d|o)\mathbb{P}(r|o,d).
\end{align}
The equality holds because probability density functions (pdf) have unit integral. Hence,
\begin{align}\label{eq:131}
    \ln f(o,d,r) = \ln\mathbb{P}(o) + \ln\mathbb{P}(d|o) + \ln\mathbb{P}(r|o,d).
\end{align}

If we are only estimating the route choice model and assume that travel times are fixed and given, then maximizing the log-likelihood (\ref{eq:131}) corresponds to maximizing the classical log-likelihood of the route choice models. 

Furthermore, we note that it is straightforward to define a distribution over destinations even if we only observe the origin and the travel time. By definition, the probability that the destination of an observation $(o, t)$ is $d$ is
\begin{equation}
    \mathbb{P}(d|o,t) = \mathbb{P}(o,d,t)/\mathbb{P}(o,t).
\end{equation}
The term $\mathbb{P}(o,t)$ is a normalization term as it is constant with respect to the destination. 

The mixture~\eqref{eqs:mixture} can conveniently be used to model observations with partial path information. Let $r'$ be the observed part of a partially observed path, the marginalized likelihood is then
\begin{align}
    f(o,d,r',t) 
    &= \sum_{r\in R(o,d)}\mathbb{P}(o)\mathbb{P}(d|o)\mathbb{P}(r,r')f(t;r) \\
    &= \sum_{r\in R(o,d)}\mathbb{P}(o)\mathbb{P}(d|o)\mathbb{P}(r)\mathbb{P}(r'|r)f(t;r) \\
    &= \mathbb{P}(o)\mathbb{P}(d|o)\sum_{r\in R(o,d)}\mathbb{P}(r)\mathbb{P}(r'|r)f(t;r) \\
    &= \mathbb{P}(o)\mathbb{P}(d|o)\mathbb{E}_{r}[\mathbb{P}(r'|r)f(t;r)].
\end{align}
Note that $\mathbb{P}(r'|r)$ is either one or zero, and if possible, samples should be generated in a way to guarantee that no sample gets rejected, i.e., $\mathbb{P}(r'|r)$ should not be zero. In case only the path choices are of interest, the objective is to maximize
\begin{equation}\label{eq:partial}
    \ln\mathbb{E}_{r}[\mathbb{P}(r'|r)].
\end{equation}
We further discuss partially observed paths in Section~\ref{section:solution}.

It is \rev{also} possible to combine different log-likelihoods. For example, for an observation $(o_1,d_1,t_1)$ and an observation $(o_2,d_2,r_2,t_2)$, the total log-likelihood is 
\begin{equation}
    \ln f(o_1,d_1,t_1) + \ln f(o_2,d_2,r_2,t_2).
\end{equation}
In other words, we have a formulation that can mix different types of data consistently. In contrast, using a standard weighting of losses (one loss for each data type) would introduce hyperparameters to control the ratio of the losses. Using other multi-objective optimization schemes would involve similar challenges.

As discussed above, mixture~(\ref{eqs:mixture}) allows to model observations at various levels of granularity. Next, we discuss other properties and show that the approach in \cite{bertsimas2019travel} is a special case. Inspired by the literature on energy-based models (EBM) \citep{lecun2007energy}, we aim to convert a class of loss functions to distributions whose log-likelihood is the original loss function. However, unlike the EBM literature, we do not require the loss function to be an energy function. 

\begin{prop}\label{prop:loss2dist}
For any function $L:\mathbb{D}\times\Theta\rightarrow\mathbb{R}$, defined over the observation domain $\mathbb{D}$ and \rev{the parameter space $\Theta$, let, for a parameter vector $\theta$ in $\Theta$},
\begin{equation}
    Z_L(\theta) = \int_{\mathbb{D}} \exp[-L(x, \theta)]\mathrm{d}x
\end{equation}
be the partition function. Whenever $Z_L(\theta)$ converges, 
\begin{equation}
    f_L(x;\theta) = \begin{cases}
    \exp[-L(x, \theta)] / Z_L(\theta) & \text{if $x\in\mathbb{D}$} \\
    0 & \text{otherwise,}
    \end{cases}
\end{equation}
is a valid pdf. Furthermore, if $Z_L$ is a constant independent of $\theta$, the log-likelihood of the distribution corresponding to $f_L$ is the original loss function.
\end{prop}
\textbf{Proof.} The density is necessarily positive as the exponent of any real number is positive. The integral of $f$ over all reals equals the integral over $\mathbb{D}$; we can then take out the partition function because it is constant with respect to $x$ and get the partition function divided by itself, which equals one. The log-likelihood is $\log f_L(x;\theta) = -L(x, \theta) - \log Z_L(\theta)$ hence maximizing the log-likelihood is equivalent to minimizing the loss function if $Z_L$ is a constant. $\hfill \square$

\begin{remark}
There are many examples of such functions. The mean squared error, $\mathrm{MSE}(x,y)=(x-y)^2$, corresponds to the normal distribution with $\sigma^2=0.5$ and the linear exponential function, $\mathrm{LINEX}(x,y)=\exp(x-y)-(x-y)-1$, corresponds to the Gumbel distribution \citep{atiyah2020fuzzy}. 
\end{remark}
\begin{remark}
Whereas from an optimization perspective multiplying by a number does not change the minimum of the problem; it changes the corresponding distribution.
\end{remark}

\rev{We note that while the partition function may not always be constant, it can sometimes be made constant using the following trick:}
\rev{\begin{prop}\label{prop:N}
Let $L(x,\theta)=L'(\log x, g(\theta))$ for some function of $\theta$, assuming the domains are compatible, if $Z_{L'}$ is constant then $Z_{L_s}$ where $L_s(x, \theta) = L(x, \theta) + \log(x)$ is constant.
\end{prop}}

\noindent \textbf{Proof.} 
\begin{align}
\int_\mathbb{D} \exp[-L_s(x, \theta)] \mathrm{d}x &= 
\int_\mathbb{D} \exp[-L(x, \theta) - \log x] \mathrm{d}x \\
&= \int_\mathbb{D} \exp(-L'(\log x, g(\theta)) - \log(x)) \mathrm{d}x \\
\intertext{Let $y=\log x$, we have: $\mathrm{d} y = \dfrac{\mathrm{d}x}{x}$ or $\exp(y)\mathrm{d} y = \mathrm{d}x$}
&= \int_{\log \mathbb{D}} \exp[-L'(y, g(\theta)) - y]\exp(y)\mathrm{d} y \\
&= \int_{\log \mathbb{D}} \exp[-L'(y, g(\theta))] \mathrm{d} y  \\
&= Z_{L'}
\end{align}
which is constant by construction of $L'$.$\hfill \square$
\qed

Unfortunately, the partition function resulting from applying Proposition~\ref{prop:loss2dist} to the MSLE is not constant. However, we introduce the small-time biased MSLE (SMSLE):
\begin{equation}
    \mathrm{SMSLE}_\gamma(t,h) = \gamma\mathrm{MSLE}(t,h) + \ln t.
\end{equation}
The SMSLE has the same derivative as the MSLE with respect to the parameter $h$ but has a constant partition function. We use this function in our experimental results. 

The next proposition links the work of \cite{bertsimas2019travel} to maximum likelihood estimation.

\begin{prop}
The solution to the BDJM model is the maximum likelihood estimate of (\ref{eqs:mixture}) where the route choice model is the shortest path algorithm and $\mathcal{O}$ and $\mathcal{D}(o)$ are not estimated.
\end{prop}
\noindent \textbf{Proof.}
When $\mathcal{P}$ is the shortest path, the expectation in (\ref{eq:ll}) is equal to $f(t;r)$ for the shortest path $r$. If we chose $\mathcal{Y}$ to be the log-normal, $\ln f(t;r)$ is proportional to the MSLE. As such the result of the BDJM is the maximum likelihood estimate. $\hfill \square$

While it might seem possible to address the interdependence issue with observations \rev{that include} both path and travel-time information, \rev{we note that the gradients can differ. Indeed, the gradient of $\mathbb{P}$ with respect to $\arctt$ may be nonzero. It is unclear how the wrong gradient interacts with non-convexities and model misspecification.}

\subsection{Solution Approach\label{section:solution}}

Using the approach presented in the previous subsection, we create a log-likelihood function $L$ tailored to our not necessarily homogenous dataset. We then maximize the log-likelihood:
\begin{maxi}
{\arctt\in\mathbb{T},\pmbeta}{L(\arctt,\pmbeta)}{}{\label{prog:ours}}
\end{maxi}
As we show below\rev{,} $\mathbb{T}$, the domain of $\arctt$, is often a rectangular constraint meaning the projection on it is not costly to compute. For optimizing (\ref{prog:ours}), we assume that the function $\ln\mathbb{P}(r\sim\mathcal{P})$ is differentiable with respect to the travel time, and the route choice model's parameters. 
In the following, we focus on the hardest case where observations do not have accompanying path data as having the path removes the main source of complexity, namely the marginalization over the unknown.

\begin{assumption}
    We assume that $\nabla_{\pmbeta,\arctt} \ln\mathbb{P}(r)$ exists.
\end{assumption}

The gradient of the log-likelihood~(\ref{eq:ll}) is
\begin{equation}\label{eq:fodt}
\nabla \ln f(o,d,t) = \nabla\ln\mathbb{P}(o) + \nabla\ln\mathbb{P}(d|o) + \nabla\ln\mathbb{E}_{r\sim\mathcal{P}}\left[f(t;r)\right].
\end{equation}
\rev{We are interested in expressing the gradient of $\ln \mathbb{E}_{r\sim\mathcal{P}}[f(t;r)]$ as an expectation to estimate it with samples. We call estimators that use samples from $\mathcal{P}$ \emph{online}, and those that use samples from a different distribution \emph{offline}.}

\begin{prop}
The offline estimator 
\begin{equation}
    \mathbb{E}_{r\sim\mathcal{P}'}\left[
 \nabla f(t;r) +
 f(t;r)\frac{\nabla\mathbb{P}_\mathcal{P}(r)}{\mathbb{P}_{\mathcal{P}'}(r)}\right]
\end{equation}
and online estimator
\begin{equation} \label{eq:onlineestim}
    \mathbb{E}_{r\sim\mathcal{P}}\left[
 \nabla f(t;r) + f(t;r)\nabla\ln\mathbb{P}_\mathcal{P}(r)
 \right],
\end{equation}
estimate $\nabla\ln\mathbb{E}_{r\sim\mathcal{P}}\left[f(t;r)\right]$ in (\ref{eq:fodt}) for any distribution $\mathcal{P}'$ whose support is a superset of the support of $\mathcal{P}$, i.e. such that $\mathcal{P}\ll \mathcal{P}'$.
\end{prop}
\textbf{Proof.} We rewrite the gradient of the expected likelihood as
\begin{align}
 &\nabla\mathbb{E}_{r\sim\mathcal{P}}\left[f(t;r)\right]  \\&= \nabla\mathbb{E}_{r\sim\mathcal{P}'}\left[f(t;r)\frac{\mathbb{P}_\mathcal{P}(r)}{\mathbb{P}_{\mathcal{P}'}(r)}\right] \\
 &= \mathbb{E}_{r\sim\mathcal{P}'}\left[
 \nabla f(t;r)
 \frac{\mathbb{P}_\mathcal{P}(r)}{\mathbb{P}_{\mathcal{P}'}(r)}
 +
 f(t;r)\frac{\nabla\mathbb{P}_\mathcal{P}(r)}{\mathbb{P}_{\mathcal{P}'}(r)}
 \right]\label{eq:eq233}\\
  &= \mathbb{E}_{r\sim\mathcal{P}}\left[
 \left(\nabla f(t;r)
 \frac{\mathbb{P}_\mathcal{P}(r)}{\mathbb{P}_{\mathcal{P}'}(r)}
 +
 f(t;r)\frac{\nabla\mathbb{P}_\mathcal{P}(r)}{\mathbb{P}_{\mathcal{P}'}(r)}
 \right)\frac{\mathbb{P}_{\mathcal{P}'}(r)}{\mathbb{P}_\mathcal{P}(r)} \right]\\
 &= \mathbb{E}_{r\sim\mathcal{P}}\left[
 \nabla f(t;r) +
 f(t;r)\frac{\nabla\mathbb{P}_\mathcal{P}(r)}{\mathbb{P}_{\mathcal{P}}(r)}\right]\\
&=\mathbb{E}_{r\sim\mathcal{P}}\left[
 \nabla f(t;r) + f(t;r)\nabla\ln\mathbb{P}_\mathcal{P}(r)
 \right].\label{eq:online}
\end{align}
We revert back to the original measure in (\ref{eq:eq233}).  $\hfill \square$

\begin{remark}
This type of gradient of expectation is called the REINFORCE \citep{williams1992simple} or score function estimator \citep{l1990unified,gradient}. 
Furthermore, it is possible to use the same estimator for higher order gradient as shown by \citet{gradient}.
\end{remark}

Next, we show how to derive \citet{mai2023estimation}'s solution to partially observed paths using gradient estimation.  If we apply the online estimator to (\ref{eq:partial}), we obtain
\begin{equation}\label{eq:partial:grad}
    \nabla \log \mathbb{E}_r[\mathbb{P}(r'|r)] = \frac{\mathbb{E}_r[\mathbb{P}(r'|r)\nabla \ln \rev{\mathbb{P}}(r)]}{\mathbb{E}_r[\mathbb{P}(r'|r)]}.
\end{equation}
Note that $\mathbb{P}(r'|r)$ simply indicates whether the observed part of the trajectory $r'$ is compatible with the actual trajectory $r$, and it is independent of the distribution of $r$. Therefore, it does not need to be included in the gradient estimator. Furthermore, $\mathbb{E}_r[\mathbb{P}(r'|r)]$ is equal to $\mathbb{P}(r')$ thus, using the Bayes' rule, (\ref{eq:partial:grad}) is equal to 
\begin{equation}
    \frac{\mathbb{E}_r[\mathbb{P}(r'|r)\nabla \ln \mathbb{P}(r)]}{\rev{\mathbb{P}}(r')} = \frac{\sum_{r \in R(o,d)} \mathbb{P}(r) \mathbb{P}(r'|r) \nabla \ln \mathbb{P}(r)}{\rev{\mathbb{P}}(r')} = \sum_{r \in R(o,d|r')} \mathbb{P}(r|r') \nabla \ln \mathbb{P}(r),
\end{equation}
where $R(o,d|r')$ is the set of all paths from $o$ to $d$ that contain (are compatible with) $r'$. As a result, (\ref{eq:partial:grad}) is equal to $\mathbb{E}_{r|r'}[\nabla \ln\mathbb{P}(r)],$
which is the gradient of Equation (9) in \citet{mai2023estimation}.

\subsection{\rev{A Worked Example}} \label{sec:workedExample}

\begin{figure}
    \centering
    \includegraphics[width=0.5\linewidth]{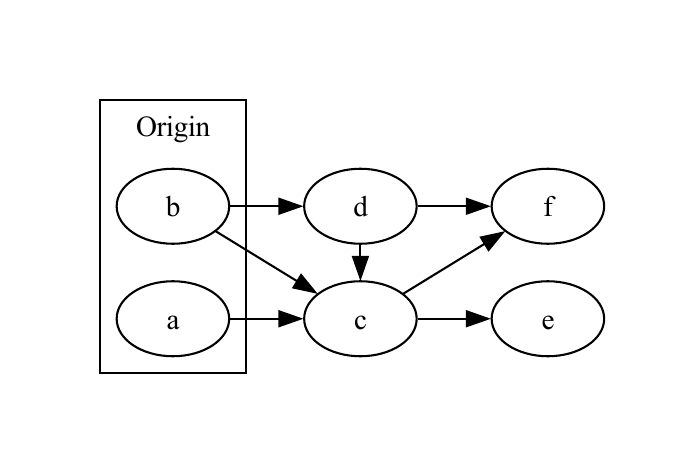}
    \caption{\rev{A simple network with a trajectory that originates at a or b, pass through c, and end at e.}}
    \label{fig:placeholder}
\end{figure}

\rev{To illustrate the previous ideas regarding different granularities, we use the simple graph depicted in Figure~\ref{fig:placeholder}. Suppose that we have a single observation that a driver started in the origin zone (at nodes a or b), passed through c and ended at e. Furthermore, suppose that we observe the travel time, and that our travel time distribution has full rank so that travel time does not impose any path limitation. Then, the set of paths compatible with the observations is 
\begin{enumerate}
    \item a\textrightarrow c\textrightarrow e,
    \item b\textrightarrow c\textrightarrow e, and
    \item b\textrightarrow d\textrightarrow c\textrightarrow e.
\end{enumerate}}

\rev{The likelihood of the observation is thus }

\rev{$ \mathbb{P}(a)\mathbb{P}(e|a)\mathbb{P}(a,c,e)f(t;r(a,c,e)) +
\mathbb{P}(b)\mathbb{P}(e|b)(
\mathbb{P}(b,c,e)f(t;r(b,c,e)) +
\mathbb{P}(b,d,c,e)f(t;r(b,d,c,e))).
$}
\rev{However, for real networks it is not possible to enumerate all compatible paths. Instead, sampling schemes and gradient estimators become a vital part of dealing with unobserved variables. This may require rejection sampling, as it is not always obvious how to efficiently sample with constraints. Still, in the logit case, constraining the sampling with an intermediate node is equivalent to sampling two paths -- one from the origin to the intermediate node and one from the intermediate node to the destination -- and concatenating them.}

\rev{This view can also be extended to handle noisy observations, for instance, given an uncertainty set around the observations, we can account for any missing data directly. Suppose that we had observed a path going from $a$ to $c$ to $e$ but the first observation was noisy, such that we are unsure whether the first observation was actually $a$ or $b$, we can use a noise model to populate $\mathbb{P}(a|\text{having observed $a$})$ and $\mathbb{P}(b|\text{having observed $a$})$.
}

\subsection{Details}

In the remainder of this section, we describe several important aspects of our solution approach.

\paragraph{Inference.} Inference depends on the performance metric. If the performance metric is likelihood, then we maximize the likelihood (i.e., find the mode). If the performance metric is the MSLE, we optimize the expected MSLE between the prediction and samples from our mixture. Indeed, the empirical loss converges to the actual expected MSLE as the number of observations tends to infinity due to the law of large numbers.

To simplify the inference process, we replace the expectation over paths with the empirical expectation over samples. To calculate the mode, we use a grid search. \citet{bertsimas2019travel} showed that for any distribution, the geometrical mean minimizes the MSLE. To minimize the MSLE on any path, we use the law of total expectation to calculate the geometrical average as
\begin{align}
\exp\left(\mathbb{E}_{r\sim\mathcal{P},t\sim\mathcal{Y}(h)}[\ln t]\right) = \exp\left(\mathbb{E}_{r\sim\mathcal{P}}\left[
\mathbb{E}_{t\sim\mathcal{Y}(h|r)}[\ln t|r]\right]\right).
\end{align}

\paragraph{Domain of $\arctt$.} We restrict the travel time estimation to a reasonable domain. For instance, the restriction can be as simple as limiting the arc travel times by bounding the speed to be between the speed limit and a lower bound. Furthermore, the domain can encode restrictions from different types of observations, such as sensors that measure the average time to traverse a specific road segment. 

We use projected gradient-based methods to maximize the log-likelihood. Projection on the rectangular (per arc) domain is inexpensive and does not require explicitly solving the Karush–Kuhn–Tucker (KKT) conditions.

\paragraph {Identifiability.} The model may not be identifiable. For instance, when maximizing (\ref{eq:marginalized_t}), i.e., travel time estimation without travel time observation, the model is undetermined if we use \rev{discrete choice} model. Indeed, we can multiply the utility parameter associated with travel time by a constant and divide all travel times by the same constant without changing the likelihood. Colinearities can generally create a large class of solutions with a very close likelihood. The bounds on $\arctt$ help mitigate this issue. 

\paragraph{Regularization.} We propose the following differentiable regularization instead of the one proposed by \citet{bertsimas2019travel} as \rev{theirs --  see (\ref{eq:Bertsimasreg}) in this paper --} is not differentiable. Our regularization is the summation of 
\begin{equation}
    \text{MSLE}(\arctt_{ij}/\arcll_{ij}, \arctt_{jk}/\arcll_{jk}) / (\arcll_{ij} + \arcll_{jk})
\end{equation}
over all 3-tuple $(i,j,k)$ such that $(i,j)$ and $(j,k)$  are arcs.

\paragraph{Optimization.} A source of complexity is the stochastic gradient when estimating (\ref{eq:fodt}). The stochastic objective makes line search and classical methods like L-BFGS-B \citep{lbfgsb} perform poorly. This is  unsurprising as the stochastic nature of the problem defies the basic assumptions of smooth optimization. Instead, we use the Adam optimizer \citep{kingma2014adam} with a projection of gradients before and a projection of variables after each update. Mini-batches can be used to improve performance, speed up convergence, and are effective for training over large datasets.

As the estimated gradient is noisy, we use a primal stopping criterion instead of using a first-order stopping criterion. We stop optimizing when we cannot improve the loss by a fixed threshold over the course of a fixed number of iterations.

\paragraph{Variance estimation and lack thereof.} It is customary to analyze the estimator's variance. However, the variance estimators \citep[p.~243]{dasgupta2008asymptotic} or ``\rev{Huber}'s sandwich estimator'' \citep{sandwich} do not have the standard interpretations in our setting for a few reasons. First, our problem is not convex, and we cannot guarantee convergence to a global optimum. Second, not all variables are free in the constrained problem. Third, the gradient we estimate is, in many cases, not exact. Last but not least, we have a large number of parameters. As such, the standard methods do not apply, and estimating the variance is beyond the scope of this work.

\paragraph{Complexities.} The methodology we describe in this section is based on relatively simple ideas. However, since it builds on top of route choice models, any complexity inherent to such models persists.
The stochasticity of the evaluation and non-convexity are other sources of complexity. \rev{Note, however, that our experiments indicate that the results are not sensitive to initialization.} By choosing a differentiable route choice model, we avoid iterative schemes over discrete variables like the BDJM method (the solution to the BDJM model) as we optimize a differentiable problem.

Writing a tractable implementation requires some care. The BDJM method transfers most of the burden of a fast implementation to a solver, as calculating shortest paths is fast. We are faced with a few implementation challenges but with enough parallelism, efficient gradient calculation using \citet{mai2018decomposition}, and efficient operations by means of sparse matrices, we are able to have a fast implementation, as we show in the following section. \rev{Finally, we note that the solution approach is agnostic to the network structure and requires only a connected graph representation with associated arc features.}

\section{Experimental Results} \label{section:experiments}

In this section, we illustrate our method using a recursive logit model \citep{fosgerau2013link}. 
We present two sets of results. In Section~\ref{sec:results}, we report results for Manhattan using the Yellow Cab dataset (NYC). The objectives are to assess the performance on the travel time estimation problem and compare it to the BDJM method. We also analyze the parameter estimates of the resulting route choice model. Since the Yellow Cab dataset does not contain path observations, we cannot use it to compare our method to the two-step approach. Indeed, the second-step -- route choice model estimation -- is not possible using existing methods. Therefore, we provide a set of results in Section~\ref{sec:synthetic} based on simulated data and compare parameter estimates to ground-truth values. We also compare results obtained with and without path observations and compare our method to the two-step approach. Next, we describe the experimental setup in more detail.

\subsection{Experimental Setup}

We split the data into three datasets: training, validation, and testing. We report the square root of MSLE (RMSLE) and we use NOMAD \citep{nomad} to search the hyperparameter space by optimizing the validation loss. In this way we found hyperparameter values with better performance compared to those recommended in \citet{bertsimas2019travel}. For the BDJM method, we use the MOSEK solver \citep{mosek} with their proposed termination criterion and a maximum of 30 iterations. We compared the performance of MOSEK, Gurobi, CPLEX and SCS and chose MOSEK because it was faster than the others.%

Our method uses the primal stopping criterion to terminate after 50 iterations with less than 0.01 total improvement in objective. We employ the online estimator (\ref{eq:onlineestim}) as it constantly outperforms the offline estimator. We use 35 samples for NYC and 100 for the synthetic set to estimate $\mathbb{E}_r[f(t;r)]$.

There are five features in our recursive logit route choice model. The first two relate to arc travel time in minutes: We partition the travel time matrix into two disjoint matrices such that the sum of those two matrices equals the original one, and the elementwise multiplication produces a zero matrix. The travel times for secondary streets are stored in the second matrix and the rest in the first. The third feature indicates red lights, stops, or intersections. Lastly, we have one indicator for left turns and one for U-turns. We can add left and U-turns by extending the current state with the knowledge of the previous state, i.e., arcs become nodes in a lifted graph. We design the features such that discretization does not affect path probabilities. We fix the value of the $\beta_\text{U-turn}$ to -5, similar to \citet{fosgerau2013link}. The arc utilities are linear in the parameters.

The methods are coded in Julia \citep{julia}. The code is publicly available (we will add the GitHub link after the double-blind review process is over). More details about data processing are given in Appendix~\ref{appendix:data}.

\subsection{Synthetic Data} \label{sec:synthetic}

In this section, we analyze results on data that has been simulated on a grid network with a route choice model that we posit (see Appendix~\ref{appendix:data} for details). Since we have access to observations at different levels of granularity, we can choose to use path observations or only observed OD times (as in the Yellow Cab data). This results in different versions of our method: with/without regularization (reg.) and with/without path observations.

We start with a qualitative analysis. Figure~\ref{fig:my_labelsd} displays a visualization of $\arctt$ for the methods; the ground truth is shown in Figure~\ref{fig:my_labelsd:ground}. 
The RMSLE between the estimated $\hat{\arctt}$ and the ground truth is below 0.14 for our model, below 0.08 for our model when using path observations, and 0.22 for the BDJM method.
However, based on visual inspection, our estimates of $\arctt$ are closer to the ground truth despite all methods having a similar RMSLE (ours is 0.31 while BDJM is 0.33, here the metric is computed based on the travel time matrices). This highlights one of the weaknesses of measuring model performance with RMSLE. One alternative would be to compare predicted travel time for known paths, but it is not always possible.

\begin{figure}
    \centering
    \begin{subfigure}[t]{0.3\linewidth}
    \includegraphics[width=\linewidth]{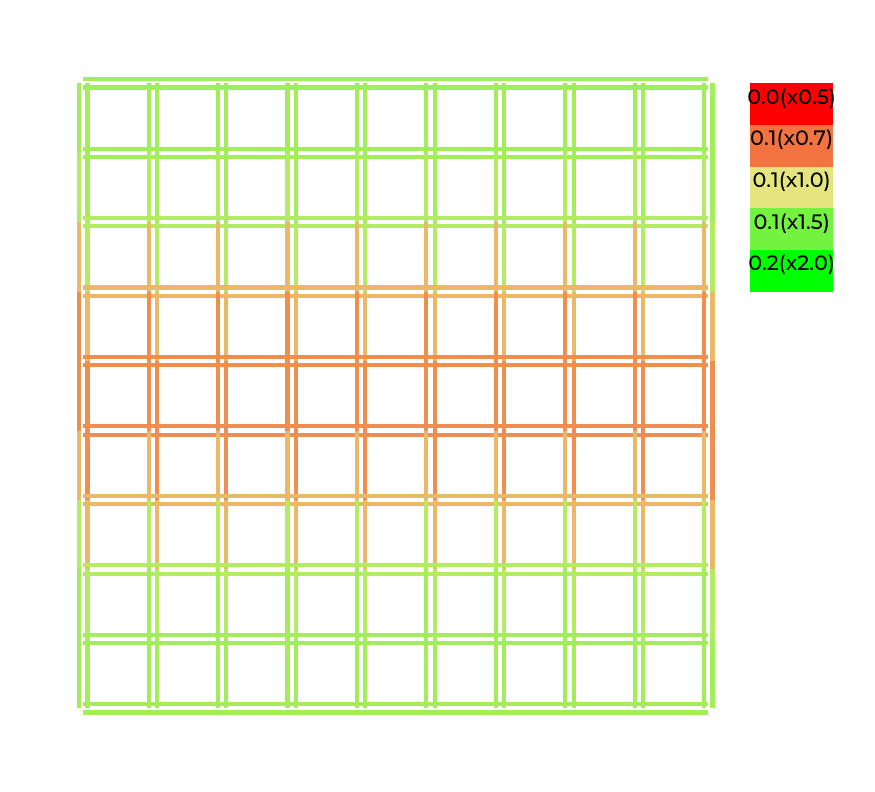}
    \caption{Ground truth}\label{fig:my_labelsd:ground}
    \end{subfigure}%
    ~\begin{subfigure}[t]{0.3\linewidth}
    \includegraphics[width=\linewidth]{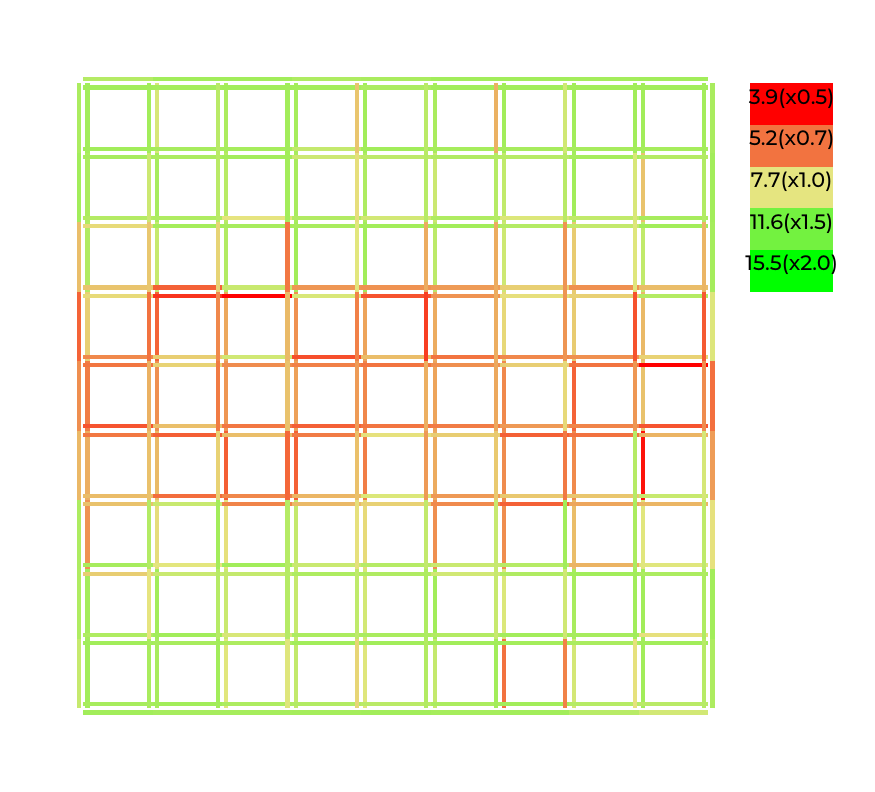}
    \caption{Ours}
    \end{subfigure}%
    ~\begin{subfigure}[t]{0.3\linewidth}
    \includegraphics[width=\linewidth]{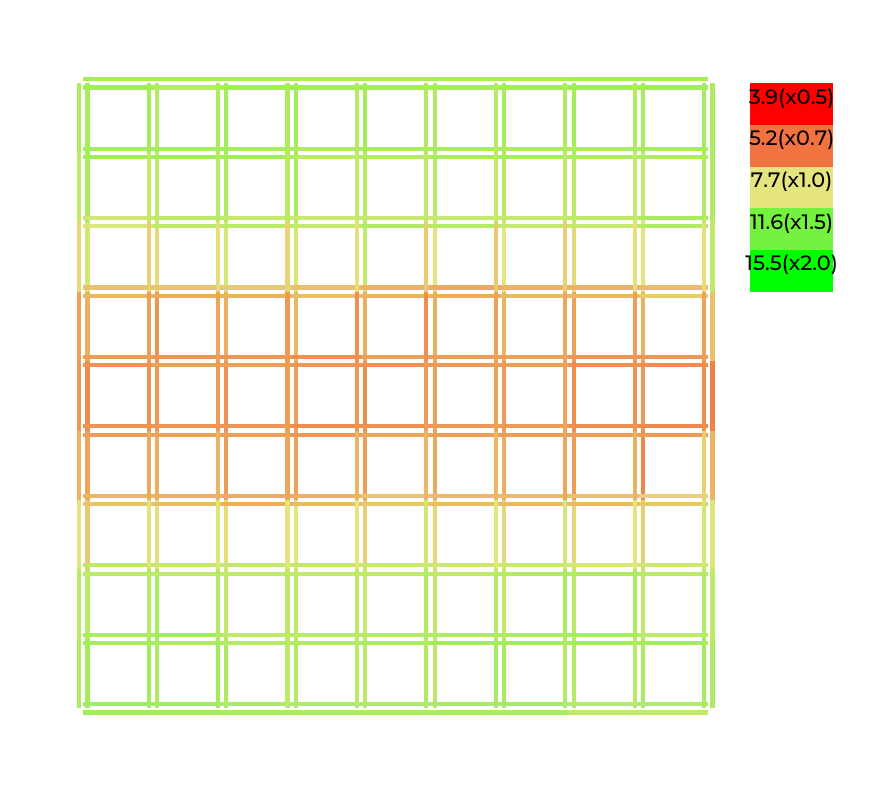}
    \caption{Ours reg.}
    \end{subfigure}
    
    \begin{subfigure}[t]{0.3\linewidth}
    \includegraphics[width=\linewidth]{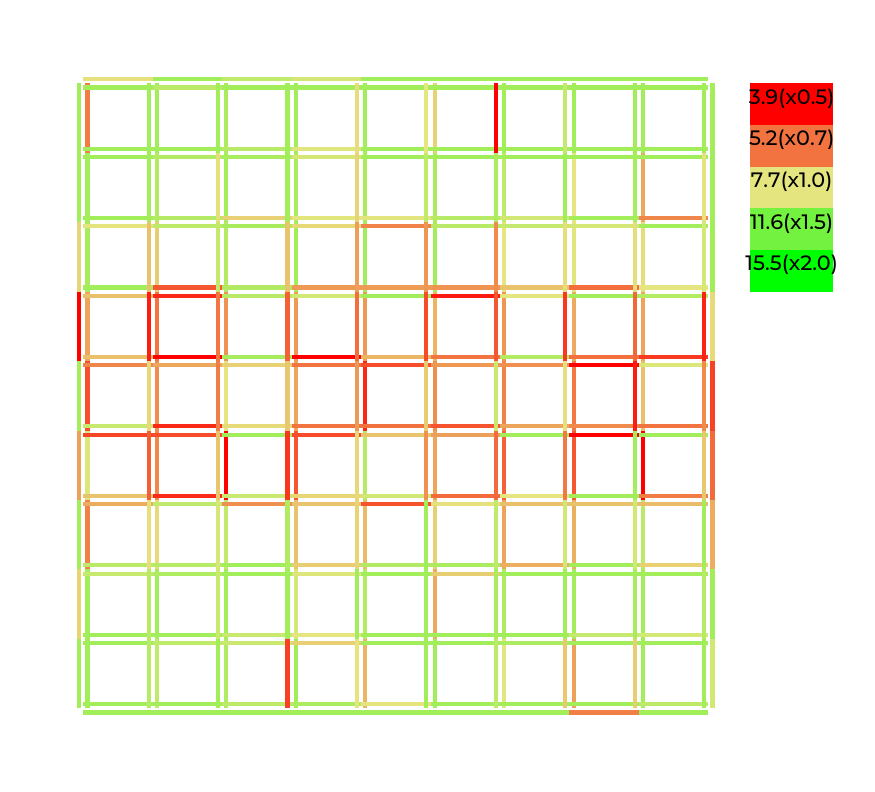}
    \caption{BDJM's method}
    \end{subfigure}%
    ~\begin{subfigure}[t]{0.3\linewidth}
    \includegraphics[width=\linewidth]{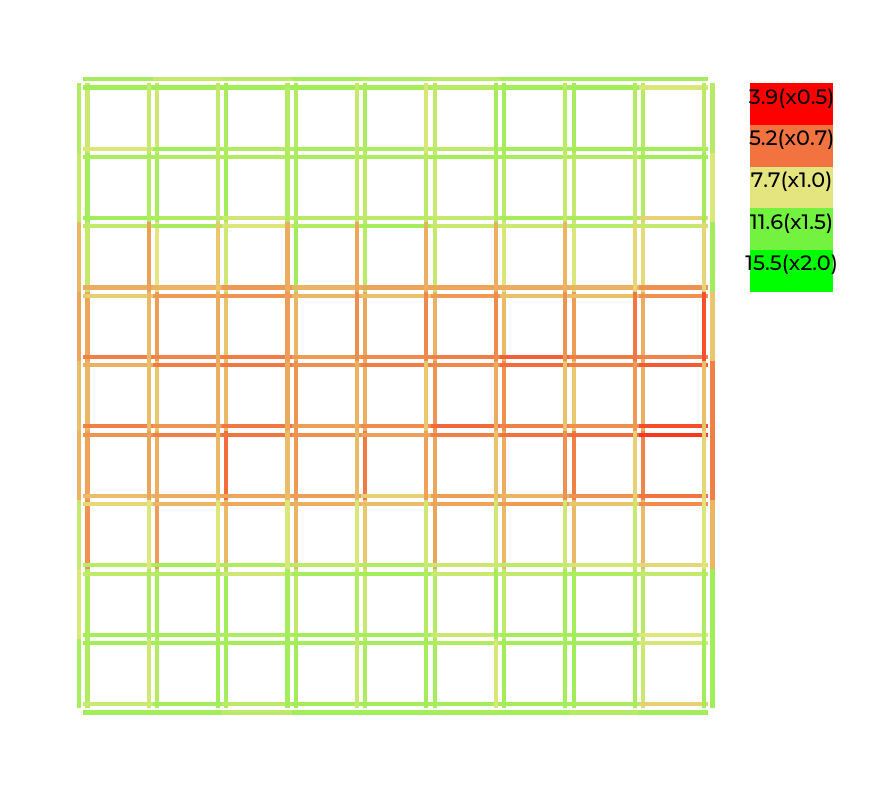}
    \caption{Ours path}
    \end{subfigure}%
    ~\begin{subfigure}[t]{0.3\linewidth}
    \includegraphics[width=\linewidth]{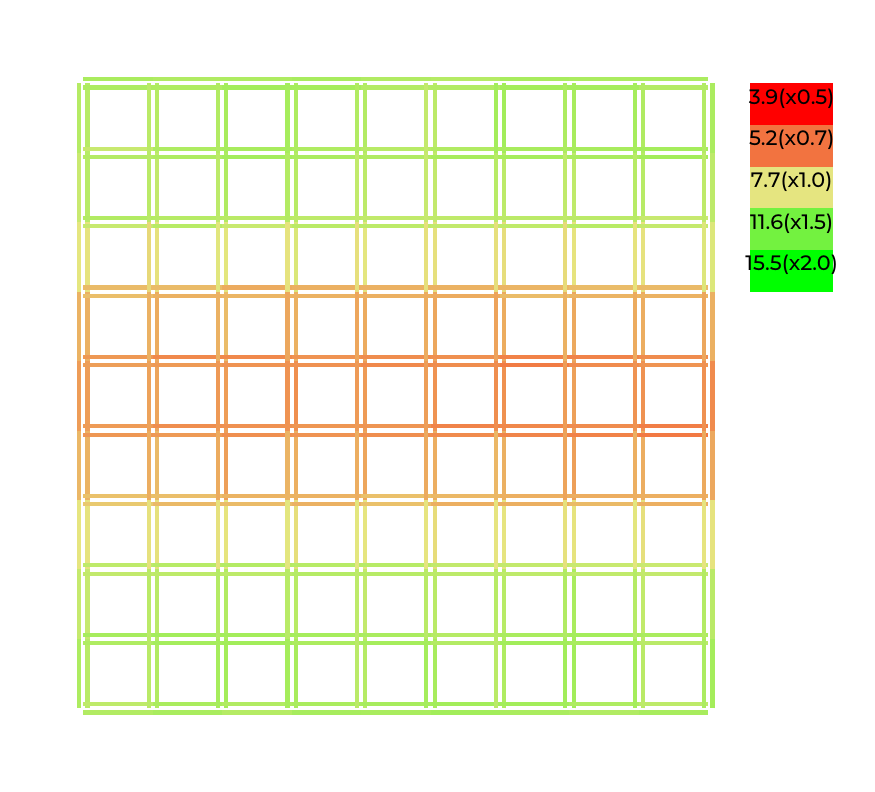}
    \caption{Ours path and reg.}
    \end{subfigure}

    \caption[Comparison of estimated arc travel time matrix.]{Comparison of estimated arc travel time matrix, same color scale for all figure. reg. stands for regularized.}
    \label{fig:my_labelsd}
\end{figure}

Next, we turn our attention to an analysis of the parameter estimates. Here the main objective is to compare our method's estimates to ground truth values, and those obtained using the two-step method, i.e., estimating a route choice model assuming travel time is fixed and given. For the two-step method, we use travel times calculated assuming 90\% of free-flow speed or the BDJM method in the first step. 

The results are reported in Table~\ref{tab:york:beta}. Note that we compare parameter ratios. The two-step method cannot retrieve the ground truth parameter ratios, irrespectively of how the travel times are estimated. Given that recursive logit has been used to generate the data, the path choice assumption in the first step is misspecified for the free-flow method and BDJM's method. The \rev{inaccuracies} in the travel time \rev{estimates are} absorbed by the travel time parameter estimate. On the contrary, our method achieves parameter estimate ratios close to the ground truth ratio. While $\pmbeta$ is well estimated in all cases, the methods using observed paths have a more accurate estimate of $\arctt$. Finally, we note that, unsurprisingly, having path observations makes the problem deterministic and easier to solve.

\begin{table}[H]
\caption[Comparison of log-likelihood.]{Comparison of log-likelihood and parameter estimates on simulated data. \label{tab:york:beta}}
    \centering\makebox[\textwidth]{
    \begin{tabular}{l
    S[table-format=-1.2]
    S[table-format=-1.2]
    }
    \toprule
    Model name & 
    {log-likelihood} &
    {$\hat{\pmbeta}_\text{travel time}/\hat \pmbeta_\text{left turn}$} \\
    \midrule
    Ground truth & -1.65 & 1.00 \\
    \midrule
    Ours                & -1.83 & 1.01 \\
    Ours with reg.      & -1.70 & 1.07 \\
    Ours with path      & -1.67 & 0.96 \\
Ours with path and reg. & -1.66 & 1.01 \\ \midrule
    Two-step (1st step: 90\% of free-flow speed)    & -2.14 & 0.78 \\
    Two-step (1st step: BDJM method)   & -1.88 & 0.72 \\
    \bottomrule
    \end{tabular}}

\end{table}

\subsection{Yellow Cab Dataset} \label{sec:results}

In Table~\ref{tab:new_york_trip}, we report performance metrics -- the number of iterations, computing time in minutes, and RMSLE --  of our and the BDJM methods for different sizes of the training set.  The first set of rows (2-6) is limited to the 6~am to 9~pm period on weekdays, whereas the second set of rows (7-11) is for the 9~am to 12~am period. The iteration budget (It.) is the training time in minutes that leads to stopping each iteration of the hyperparameter optimization loop. The algorithms may take more than It. minutes as we let the last iteration finish. We evaluate each model at most 50 times in the hyperparameter search loop. The training time columns (t.) report the average training times in minutes for the best hyperparameters. We run each experiment three times and we report the number of runs, out of the three, that converge before the time limit (\#c/3). We acknowledge that running experiments only three times is not enough to get a statistically significant variance and mean. However, there is very little variation between the runs. We note that the BDJM method and MOSEK are deterministic (to be more precise, the BDJM method is only deterministic up to tie-breaking in case of multiple shortest paths with the same length, and MOSEK is deterministic as long as the number of cores does not change), whereas our method and NOMAD depend on the seed. Due to negligible variance (less than 5e-4 in most cases), we run the BDJM method only once. The BDJM method \rev{reached} timeout in the larger set. However, the loss change in these cases is in the order of 1e-4 in the last two or three iterations, and training for longer periods did not improve the results.

\begin{table}[htbp]
\caption[Comparison with BDJM methods on the travel time estimation.]{Comparison of our and BDJM methods on the travel time estimation task on Manhattan on March 2016.}
\label{tab:new_york_trip} 
\centering
\makebox[\textwidth]{\begin{tabular}{
rr
S[table-format=1.4(2)]
r
r
S[table-format=1.4(2)]
r
r
S[table-format=1.4]
r
r}
\toprule
{size} & It. & 
{RMSLE} & {t.} & {\#c/3} &
{RMSLE} & {t.} & {\#c/3} &
{RMSLE} & {t.} & {\#c/1}  \\
{training set} & {[min]} &
{ours} & {[min]}&  &
{ours with reg.} & {[min]}&  &
{BDJM} & {[min]} &\\
\midrule
100 & 4 & 0.4343 \pm 0.0039 &  4 & 1 & 0.4051 \pm 0.0047 &  3& 2 & 0.3705 &  1 & 1\\
1,000 & 8 & 0.3918 \pm 0.0038 &  9 & 0 & 0.3891 \pm 0.0075 &  6& 2 & 0.3557 &  2 & 1\\
10,000 & 16 & 0.3472 \pm 0.0017 &  7 & 3 & 0.3485 \pm 0.0009 &  4 & 3 & 0.3571 &  18 & 0\\
100,000 & 32 & 0.3291 \pm 0.0005 &  30&2 & 0.3289 \pm 0.0004 &  31&2 & 0.3310 &  41&0\\
200,000 & 64 & 0.3264 \pm 0.0001 &  51&2 & 0.3267 \pm 0.0002 &  52&2 & 0.3285 &  80&0\\
\midrule
100 & 4 & 0.3972 \pm 0.0129 &  5&0 & 0.3900 \pm 0.0114 &  4&1 & 0.3532 &  1 & 1\\
1,000 & 8 & 0.3570 \pm 0.0015 &  9&0 & 0.3358 \pm 0.0047 &  6 & 3 & 0.3306 &  2 & 1\\
10,000 & 16 & 0.3120 \pm 0.0002 &  16&1 & 0.3113 \pm 0.0011 &  14&2 & 0.3092 &  15 & 1\\
100,000 & 32 & 0.2936 \pm 0.0001 &  30& 3 & 0.2940 \pm 0.0001 &  30&2 & 0.2983 &  41 & 0\\
200,000 & 64 & 0.2917 \pm 0.0002 &  57&2 & 0.2916 \pm 0.0001 &  54 & 3& 0.2949 &  83 & 0\\
\bottomrule
\end{tabular}}
\end{table}

In terms of training time, our method scales well with the size of the training set, taking around 30 minutes to train with 100,000 observations, roughly 10 minutes less than the BDJM method. Our method has poorer performance on the smaller datasets (100 and 1,000 observations) and we note that regularization plays an important role in data sparse settings and helps reduce the variance of the estimated $\arctt$ among experiments. Furthermore, regularization seems to help with overfitting in smaller datasets. For the larger sets (100,000 or 200,000), regularization, however, deteriorates performance (RMSLE), and NOMAD sets the regularization very low or to zero. We also note that even though our method is stochastic, the variance (in RMSLE) is relatively low. We note that in the larger datasets we often are able to terminate the algorithm within the time limit and that regularization helps with that, specially in data poor regimes. 

Next, we focus on route choice model parameter estimates resulting from our method. For the sake of illustration, we report those obtained based on 200,000 observations and the late morning data in Table~\ref{tab:beta}. We note that there is some variation between the three different runs, but the parameter ratios remain relatively stable. Noteworthy is the difference between residential and non-residential roads, where the taxis prefer the latter.

\begin{table}[htbp]
\caption[Parameters $\pmbeta$ estimated on NYC.]{Parameters $\pmbeta$ estimated on NYC using 200,000 observations from the late morning dataset.\label{tab:beta}}
\centering
\begin{tabular}{ccccc}
    \toprule 
$\hat \pmbeta_\text{non-residential}$ & 
$\hat \pmbeta_\text{residential}$ &
$\hat \pmbeta_\text{intersection}$ &
$\hat \pmbeta_\text{left turn}$ &
$\pmbeta_\text{U-turn}$ (fixed) 
\\\midrule
-3.27 & -3.88 & -0.36 & -0.43 & -5\\
-3.50 & -4.50 & -0.49 & -0.54 & -5\\
-3.39 & -4.38 & -0.32 & -0.22 & -5
\\\bottomrule
\end{tabular}
\end{table}

With our approach we tackle the more challenging task of estimating both route choice model parameters and arc travel times. Given the shortest path assumption in the BDJM method one might expect a larger improvement in RMSLE for our method compared to BDJM. In the following we analyze why the quality, as measured by RMSLE, is similar on this dataset. 
For this purpose we pseudo randomly sample 10,000 OD pairs and 300 paths between each OD pair and we compute the ratio between the travel time of each sampled path to the corresponding shortest path travel time. Figure~\ref{fig:sp_vs_sample_hist} displays the cumulative distribution of this ratio. The two plots display the same information, but the plot on the left-hand side is a zoom (x-axis range from 1.000 to 1.100). These results show that 93.3\% of sampled paths are less than 10\% longer than the shortest path. In other words, the estimated route choice model confirms that the shortest path is a good proxy for taxi data. \rev{However, compared to a deterministic shortest-path model minimizing travel time only, the route choice model indicate that} taxi drivers in the data penalize travel time on secondary (residential) roads more severely than travel time on primary (non-residential) roads.

\begin{figure}[htbp]
    \centering
    \begin{subfigure}{0.45\linewidth}
        \includegraphics[width=\linewidth]{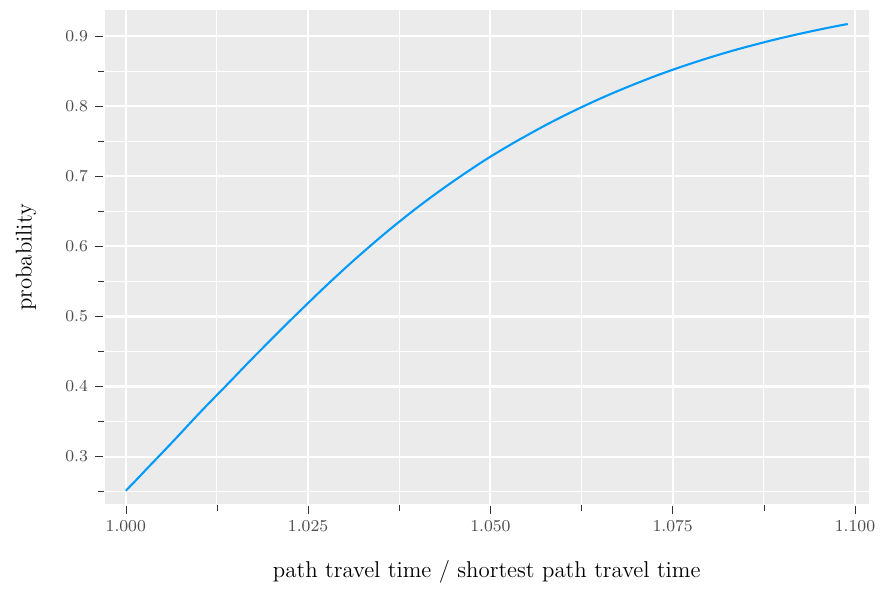}
    \end{subfigure}
    \begin{subfigure}{0.445\linewidth}
        \includegraphics[width=\linewidth]{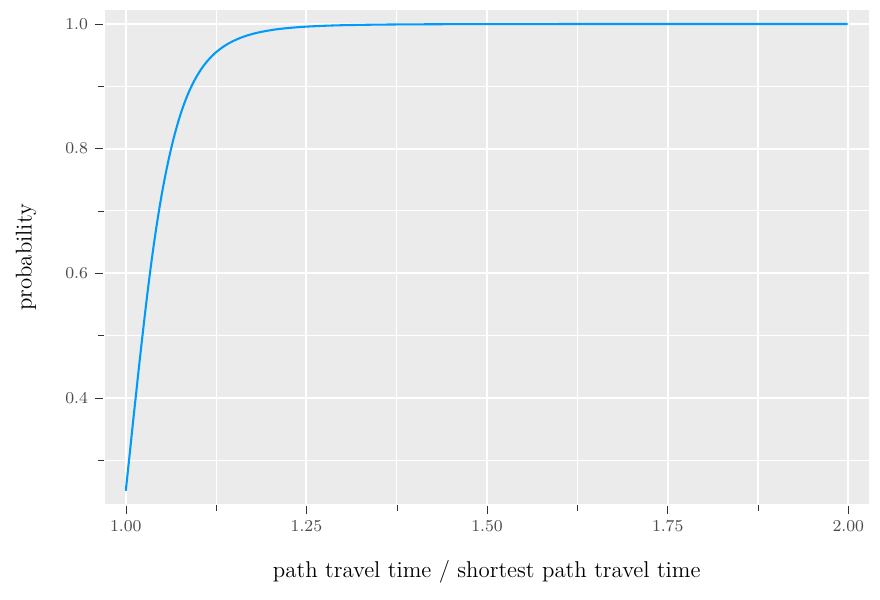}
    \end{subfigure}
    
    \caption[Comparison of path travel time and shortest path.]{Cumulative distribution of the ratio of sampled paths and shortest path travel times.}
    \label{fig:sp_vs_sample_hist}
\end{figure}

We end this section with a few remarks. Our experiments showed little sensitivity to the initialization. The effects of regularization depend on the data-set size; the more data, the less regularization is required. As expected, low learning rates lead to more steps in the optimization loops (some results are shown in Appendix~\ref{appendix:fig}). Finally we note that initialization does not have a big impact on our results. Nevertheless, we propose the parameters $\pmbeta_0=-2$, $\arctt_0=\arctt_\text{min}/0.9$, $\log\eta=-2$ and $\lambda=0$. Another viable alternative is to initialize using the solution of the two-step method and then use our method to improve the estimation results.

\section{Conclusion} \label{sec:Conclusion}

\rev{The mixture we proposed} models traffic for tactical and strategic network planning, \rev{while} ensuring compatibility with a \rev{wide} array of route choice models and loss functions. It is designed for simultaneous estimation of arc travel time and route choice model parameters. Our approach underscores the advantage of marginalizing unobserved variables and utilizing stochastic gradient estimates, leading to a maximum likelihood estimation, even when observations occur at different levels of granularity. Notably, we demonstrated that various data types can be amalgamated when computing the maximum likelihood estimate without necessitating a linear combination of losses as an objective. Speed is another advantage of optimizing this mixture. 

We illustrated with a small example as well as simulation that the conventional two-step approach for estimating route choice models may result in suboptimal outcomes. This underscores the value of our methodology.

Despite these promising outcomes, our method requires further exploration across additional route choice models, datasets, and different travel time distributions. As articulated in Section~\ref{section:methodology}, the sole precondition is the differentiability of the route choice model. However, further experiments are important to validate the practical effectiveness of our methodology across this wide range of route choice models. \rev{For example, modeling heterogeneous preferences using datasets also covering personal trips rather than only taxi drivers.}

Our choice of Manhattan as a test site was intended to align the experimental section \rev{with other studies using the same data, including} \citet{bertsimas2019travel}. \rev{Our methodology is flexible and operates on arbitrary graphs, requiring only connectivity and arc weights.} Still, we acknowledge the importance of assessing the model's performance in diverse locations, particularly those deviating from Manhattan's grid-like structure. Moreover, testing the method on a mix of observations at different levels of granularity remains an area for future exploration. Lastly, we note that further investigation is needed for estimating time-dependent travel times \rev{and modeling spatiotemporal correlation.}

\section*{Acknowledgments}
We will add the original section after the double-blind review process is over

\bibliographystyle{plainnat_custom}
\bibliography{refs}

\begin{thebibliography}{60}
\providecommand{\natexlab}[1]{#1}
\providecommand{\url}[1]{\texttt{#1}}
\expandafter\ifx\csname urlstyle\endcsname\relax
  \providecommand{\doi}[1]{doi: #1}\else
  \providecommand{\doi}{doi: \begingroup \urlstyle{rm}\Url}\fi

\bibitem[{Abdzaid Atiyah} et~al.(2020){Abdzaid Atiyah}, Mohammadpour, Ahmadzadehgoli, and Taheri]{atiyah2020fuzzy}
{Abdzaid Atiyah}, I., Mohammadpour, A., Ahmadzadehgoli, N., and Taheri, S.
\newblock Fuzzy {C}-means clustering using asymmetric loss function.
\newblock \emph{Journal of Statistical Theory and Applications}, 19\penalty0 (1):\penalty0 91--101, 2020.

\bibitem[Agrawal et~al.(2019)Agrawal, Amos, Barratt, Boyd, Diamond, and Kolter]{agrawal2019differentiable}
Agrawal, A., Amos, B., Barratt, S., Boyd, S., Diamond, S., and Kolter, J.~Z.
\newblock Differentiable convex optimization layers.
\newblock \emph{Advances in neural information processing systems}, 32, 2019.

\bibitem[Amos and Kolter(2017)]{amos2017optnet}
Amos, B. and Kolter, J.~Z.
\newblock Optnet: Differentiable optimization as a layer in neural networks.
\newblock In \emph{International Conference on Machine Learning}, pages 136--145. PMLR, 2017.

\bibitem[Audet et~al.(2021)Audet, Digabel, Montplaisir, and Tribes]{nomad}
Audet, C., Digabel, S.~L., Montplaisir, V.~R., and Tribes, C.
\newblock Nomad version 4: Nonlinear optimization with the {MADS} algorithm.
\newblock \emph{arXiv preprint arXiv:2104.11627}, 2021.

\bibitem[Bertsimas et~al.(2019)Bertsimas, Delarue, Jaillet, and Martin]{bertsimas2019travel}
Bertsimas, D., Delarue, A., Jaillet, P., and Martin, S.
\newblock Travel time estimation in the age of big data.
\newblock \emph{Operations Research}, 67\penalty0 (2):\penalty0 498--515, 2019.

\bibitem[Bezanson et~al.(2017)Bezanson, Edelman, Karpinski, and Shah]{julia}
Bezanson, J., Edelman, A., Karpinski, S., and Shah, V.~B.
\newblock Julia: A fresh approach to numerical computing.
\newblock \emph{SIAM review}, 59\penalty0 (1):\penalty0 65--98, 2017.

\bibitem[Bierlaire and Frejinger(2008)]{BierFrej08}
Bierlaire, M. and Frejinger, E.
\newblock Route choice modeling with network-free data.
\newblock \emph{Transportation Research Part C: Emerging Technologies}, 16\penalty0 (2):\penalty0 187--198, 2008.

\bibitem[Chen et~al.(2025)Chen, Schmidt, Ma, and Sun]{ChenEtAl25}
Chen, X., Schmidt, A.~M., Ma, Z., and Sun, L.
\newblock Bayesian spatiotemporal modeling of passenger trip assignment in metro networks, 2025.
\newblock URL \url{https://arxiv.org/abs/2507.22403}.

\bibitem[Choi et~al.(2021)Choi, Kim, and Yeo]{choi2021trajgail}
Choi, S., Kim, J., and Yeo, H.
\newblock Trajgail: Generating urban vehicle trajectories using generative adversarial imitation learning.
\newblock \emph{Transportation Research Part C: Emerging Technologies}, 128:\penalty0 103091, 2021.

\bibitem[Cort{\'e}s et~al.(2023)Cort{\'e}s, Donoso, Guti{\'e}rrez, Herl, and Mu{\~n}oz]{cortes2023recursive}
Cort{\'e}s, C.~E., Donoso, P., Guti{\'e}rrez, L., Herl, D., and Mu{\~n}oz, D.
\newblock A recursive stochastic transit equilibrium model estimated using passive data from santiago, chile.
\newblock \emph{Transportation Research Part B: Methodological}, 174:\penalty0 102780, 2023.

\bibitem[DasGupta(2008)]{dasgupta2008asymptotic}
DasGupta, A.
\newblock \emph{Asymptotic theory of statistics and probability}, volume 180.
\newblock Springer, 2008.

\bibitem[de~Freitas et~al.(2019)de~Freitas, Becker, Zimmermann, and Axhausen]{de2019modelling}
de~Freitas, L.~M., Becker, H., Zimmermann, M., and Axhausen, K.~W.
\newblock Modelling intermodal travel in switzerland: A recursive logit approach.
\newblock \emph{Transportation Research Part A: Policy and Practice}, 119:\penalty0 200--213, 2019.

\bibitem[{de Moraes Ramos} et~al.(2020){de Moraes Ramos}, Mai, Daamen, Frejinger, and Hoogendoorn]{ramos2020route}
{de Moraes Ramos}, G., Mai, T., Daamen, W., Frejinger, E., and Hoogendoorn, S.
\newblock Route choice behaviour and travel information in a congested network: Static and dynamic recursive models.
\newblock \emph{Transportation Research Part C: Emerging Technologies}, 114:\penalty0 681--693, 2020.

\bibitem[Derrow-Pinion et~al.(2021)Derrow-Pinion, She, Wong, Lange, Hester, Perez, Nunkesser, Lee, Guo, Wiltshire, Battaglia, Gupta, Li, Xu, Sanchez-Gonzalez, Li, and Velickovic]{Derrow-PinionEtAl21}
Derrow-Pinion, A., She, J., Wong, D., Lange, O., Hester, T., Perez, L., Nunkesser, M., Lee, S., Guo, X., Wiltshire, B., Battaglia, P.~W., Gupta, V., Li, A., Xu, Z., Sanchez-Gonzalez, A., Li, Y., and Velickovic, P.
\newblock Eta prediction with graph neural networks in google maps.
\newblock In \emph{Proceedings of the 30th ACM International Conference on Information \& Knowledge Management}, CIKM '21, page 3767–3776, New York, NY, USA, 2021. Association for Computing Machinery.
\newblock ISBN 9781450384469.
\newblock \doi{10.1145/3459637.3481916}.
\newblock URL \url{https://doi.org/10.1145/3459637.3481916}.

\bibitem[Dial(1971)]{dial1971probabilistic}
Dial, R.~B.
\newblock A probabilistic multipath traffic assignment model which obviates path enumeration.
\newblock \emph{Transportation Research}, 5\penalty0 (2):\penalty0 83--111, 1971.

\bibitem[Ding-Mastera et~al.(2019)Ding-Mastera, Gao, Jenelius, Rahmani, and Ben-Akiva]{ding2018latent}
Ding-Mastera, J., Gao, S., Jenelius, E., Rahmani, M., and Ben-Akiva, M.
\newblock A latent-class adaptive routing choice model in stochastic time-dependent networks.
\newblock \emph{Transportation Research Part B: Methodological}, 124:\penalty0 1--17, 2019.

\bibitem[Fosgerau et~al.(2013)Fosgerau, Frejinger, and Karlstrom]{fosgerau2013link}
Fosgerau, M., Frejinger, E., and Karlstrom, A.
\newblock A link based network route choice model with unrestricted choice set.
\newblock \emph{Transportation Research Part B: Methodological}, 56:\penalty0 70--80, 2013.

\bibitem[Fosgerau et~al.(2022)Fosgerau, Paulsen, and Rasmussen]{fosgerau2022perturbed}
Fosgerau, M., Paulsen, M., and Rasmussen, T.~K.
\newblock A perturbed utility route choice model.
\newblock \emph{Transportation Research Part C: Emerging Technologies}, 136:\penalty0 103514, 2022.

\bibitem[Freedman(2006)]{sandwich}
Freedman, D.~A.
\newblock On the so-called “huber sandwich estimator” and “robust standard errors”.
\newblock \emph{The American Statistician}, 60\penalty0 (4):\penalty0 299--302, 2006.

\bibitem[Frejinger and Hewitt(2025)]{FrejHewi25}
Frejinger, E. and Hewitt, M.
\newblock Perspectives on optimizing transport systems with supply-dependent demand.
\newblock \emph{INFOR: Information Systems and Operational Research}, 63\penalty0 (3):\penalty0 405--438, 2025.

\bibitem[Gao et~al.(2010)Gao, Frejinger, and Ben-Akiva]{GaoFrejBenA10}
Gao, S., Frejinger, E., and Ben-Akiva, M.
\newblock Adaptive route choices in risky traffic networks: A prospect theory approach.
\newblock \emph{Transportation Research Part C: Emerging Technologies}, 18\penalty0 (5):\penalty0 727--740, 2010.

\bibitem[Gao(2021)]{gao2021estimation}
Gao, Y.
\newblock Estimation of tourist travel patterns with recursive logit models based on wi-fi data with kyoto city case study.
\newblock 2021.

\bibitem[Ghandeharioun and Kouvelas(2022)]{GhanKouv22}
Ghandeharioun, Z. and Kouvelas, A.
\newblock Link travel time estimation for arterial networks based on sparse gps data and considering progressive correlations.
\newblock \emph{IEEE Open Journal of Intelligent Transportation Systems}, 3:\penalty0 679--694, 2022.

\bibitem[Gilbert et~al.(2015)Gilbert, Marcotte, and Savard]{GilbertEtAl15}
Gilbert, F., Marcotte, P., and Savard, G.
\newblock A numerical study of the logit network pricing problem.
\newblock \emph{Transportation Science}, 49\penalty0 (3):\penalty0 706--719, 2015.

\bibitem[Iizuka and Hato(2020)]{iizuka2020cost}
Iizuka, T. and Hato, E.
\newblock Cost sensitive estimation methods of the rl-activity-scheduling models under disasters.
\newblock 2020.

\bibitem[Jindal et~al.(2017)Jindal, Tony, Qin, Chen, Nokleby, and Ye]{JindalEtAl17}
Jindal, I., Tony, Qin, Chen, X., Nokleby, M., and Ye, J.
\newblock A unified neural network approach for estimating travel time and distance for a taxi trip, 2017.
\newblock URL \url{https://arxiv.org/abs/1710.04350}.

\bibitem[Kingma and Ba(2014)]{kingma2014adam}
Kingma, D. and Ba, J.
\newblock Adam: A method for stochastic optimization.
\newblock \emph{arXiv preprint arXiv:1412.6980}, 2014.

\bibitem[Knies et~al.(2022)Knies, Lorca, and Melo]{knies2022recursive}
Knies, A., Lorca, J., and Melo, E.
\newblock A recursive logit model with choice aversion and its application to transportation networks.
\newblock \emph{Transportation research part B: methodological}, 155:\penalty0 47--71, 2022.

\bibitem[Koch and Dugundji(2020)]{koch2020review}
Koch, T. and Dugundji, E.
\newblock A review of methods to model route choice behavior of bicyclists: inverse reinforcement learning in spatial context and recursive logit.
\newblock In \emph{Proceedings of the 3rd ACM SIGSPATIAL International Workshop on GeoSpatial Simulation}, pages 30--37, 2020.

\bibitem[LeCun et~al.(2007)LeCun, Chopra, Hadsell, Ranzato, and Huang]{lecun2007energy}
LeCun, Y., Chopra, S., Hadsell, R., Ranzato, M., and Huang, F.~j.
\newblock \emph{{Energy-Based Models}}.
\newblock The MIT Press, 2007.

\bibitem[L'Ecuyer(1990)]{l1990unified}
L'Ecuyer, P.
\newblock A unified view of the ipa, sf, and lr gradient estimation techniques.
\newblock \emph{Management Science}, 36\penalty0 (11):\penalty0 1364--1383, 1990.

\bibitem[Mai et~al.(2018)Mai, Bastin, and Frejinger]{mai2018decomposition}
Mai, T., Bastin, F., and Frejinger, E.
\newblock A decomposition method for estimating recursive logit based route choice models.
\newblock \emph{Euro Journal on Transportation and Logistics}, 7\penalty0 (3):\penalty0 253--275, 2018.

\bibitem[Mai et~al.(2015)Mai, Fosgerau, and Frejinger]{mai2015nested}
Mai, T., Fosgerau, M., and Frejinger, E.
\newblock A nested recursive logit model for route choice analysis.
\newblock \emph{Transportation Research Part B: Methodological}, 75:\penalty0 100--112, 2015.

\bibitem[Mai et~al.(2021)Mai, Yu, Gao, and Frejinger]{mai2021routing}
Mai, T., Yu, X., Gao, S., and Frejinger, E.
\newblock Routing policy choice prediction in a stochastic network: Recursive model and solution algorithm.
\newblock \emph{Transportation Research Part B: Methodological}, 151:\penalty0 42--58, 2021.

\bibitem[Mai et~al.(2023)Mai, Bui, Nguyen, and Le]{mai2023estimation}
Mai, T., Bui, T.~V., Nguyen, Q.~P., and Le, T.~V.
\newblock Estimation of recursive route choice models with incomplete trip observations.
\newblock \emph{Transportation Research Part B: Methodological}, 173:\penalty0 313--331, 2023.

\bibitem[Mo et~al.(2023)Mo, Ma, Koutsopoulos, and Zhao]{MoEtAl23}
Mo, B., Ma, Z., Koutsopoulos, H.~N., and Zhao, J.
\newblock Ex post path choice estimation for urban rail systems using smart card data: An aggregated time-space hypernetwork approach.
\newblock \emph{Transportation Science}, 57\penalty0 (2):\penalty0 313--335, 2023.

\bibitem[Mohamed et~al.(2020)Mohamed, Rosca, Figurnov, and Mnih]{gradient}
Mohamed, S., Rosca, M., Figurnov, M., and Mnih, A.
\newblock {Monte Carlo Gradient Estimation in Machine Learning}.
\newblock \emph{Journal of Machine Learning Research}, 21\penalty0 (132):\penalty0 1--62, 2020.

\bibitem[{MOSEK ApS}(2022)]{mosek}
{MOSEK ApS}.
\newblock \emph{{MOSEK} optimization suite. Version 9.3.}, 2022.

\bibitem[{New York City Taxi and Limousine Commission}(2016)]{taxi2019}
{New York City Taxi and Limousine Commission}.
\newblock New york city taxi trip data, 2016.
\newblock Data downloaded in June 2022 from \url{https://registry.opendata.aws/nyc-tlc-trip-records-pds}.

\bibitem[{OpenStreetMap contributors}(2022)]{OpenStreetMap}
{OpenStreetMap contributors}.
\newblock {OpenStreetMap project database}, 2022.
\newblock Downloaded in June 2022 from \url{https://overpass-api.de}.

\bibitem[Osorio and Atasoy(2021)]{OsorAtas21}
Osorio, C. and Atasoy, B.
\newblock Efficient simulation-based toll optimization for large-scale networks.
\newblock \emph{Transportation Science}, 55\penalty0 (5):\penalty0 1010--1024, 2021.

\bibitem[Oyama(2017)]{oyama2017markovian}
Oyama, Y.
\newblock \emph{A Markovian route choice analysis for trajectory-based urban planning}.
\newblock PhD thesis, Doctoral Thesis, The University of Tokyo, 2017.

\bibitem[Oyama(2023)]{oyama2023capturing}
Oyama, Y.
\newblock Capturing positive network attributes during the estimation of recursive logit models: A prism-based approach.
\newblock \emph{Transportation Research Part C: Emerging Technologies}, 147:\penalty0 104014, 2023.

\bibitem[Shaygan et~al.(2022)Shaygan, Meese, Li, Zhao, and Nejad]{shaygan2022traffic}
Shaygan, M., Meese, C., Li, W., Zhao, X.~G., and Nejad, M.
\newblock Traffic prediction using artificial intelligence: Review of recent advances and emerging opportunities.
\newblock \emph{Transportation research part C: Emerging technologies}, 145:\penalty0 103921, 2022.

\bibitem[Sun and Kim(2021)]{SunKim21}
Sun, J. and Kim, J.
\newblock Joint prediction of next location and travel time from urban vehicle trajectories using long short-term memory neural networks.
\newblock \emph{Transportation Research Part C: Emerging Technologies}, 128:\penalty0 103114, 2021.

\bibitem[Sun et~al.(2015)Sun, Lu, Jin, Lee, and Axhausen]{SunEtAl15}
Sun, L., Lu, Y., Jin, J.~G., Lee, D.-H., and Axhausen, K.~W.
\newblock An integrated bayesian approach for passenger flow assignment in metro networks.
\newblock \emph{Transportation Research Part C: Emerging Technologies}, 52:\penalty0 116--131, 2015.

\bibitem[Train(2002)]{Train2002}
Train, K.
\newblock \emph{Discrete Choice Methods with Simulation}.
\newblock Cambridge University Press, Berkeley, 2002.
\newblock URL \url{https://eml.berkeley.edu/books/train1201.pdf}.

\bibitem[Wei et~al.(2022)Wei, Vaze, and Jacquillat]{WeiEtAl22}
Wei, K., Vaze, V., and Jacquillat, A.
\newblock Transit planning optimization under ride-hailing competition and traffic congestion.
\newblock \emph{Transportation Science}, 56\penalty0 (3):\penalty0 725--749, 2022.

\bibitem[Williams(1992)]{williams1992simple}
Williams, R.
\newblock Simple statistical gradient-following algorithms for connectionist reinforcement learning.
\newblock \emph{Machine Learning}, 8\penalty0 (3-4):\penalty0 229--256, 1992.

\bibitem[Woodard et~al.(2017)Woodard, Nogin, Koch, Racz, Goldszmidt, and Horvitz]{WoodEtAl17}
Woodard, D., Nogin, G., Koch, P., Racz, D., Goldszmidt, M., and Horvitz, E.
\newblock Predicting travel time reliability using mobile phone gps data.
\newblock \emph{Transportation Research Part C: Emerging Technologies}, 75:\penalty0 30--44, 2017.

\bibitem[Wright(2015)]{Wright15}
Wright, S.~J.
\newblock Coordinate descent algorithms.
\newblock \emph{Mathematical Programming}, 151:\penalty0 3--34, 2015.

\bibitem[Xu et~al.(2025)Xu, Wang, and Sun]{xu2025link}
Xu, C., Wang, Q., and Sun, L.
\newblock Link representation learning for probabilistic travel time estimation.
\newblock \emph{IEEE Transactions on Intelligent Transportation Systems}, 2025.

\bibitem[Xue et~al.(2025)Xue, Tan, Ma, and Ukkusuri]{XueEtAl25}
Xue, J., Tan, R., Ma, J., and Ukkusuri, S.~V.
\newblock Data science in transportation networks with graph neural networks: A review and outlook.
\newblock \emph{Data Science for Transportation}, 7\penalty0 (10):\penalty0 1--27, 2025.

\bibitem[Yan et~al.(2024)Yan, Johndrow, Woodard, and Sun]{YanEtAl24}
Yan, C., Johndrow, J., Woodard, D., and Sun, Y.
\newblock Efficiency of eta prediction.
\newblock \emph{SIAM Journal on Mathematics of Data Science}, 6\penalty0 (2):\penalty0 227--253, 2024.

\bibitem[Yuan et~al.(2020)Yuan, Li, Bao, and Feng]{YuanEtAl20}
Yuan, H., Li, G., Bao, Z., and Feng, L.
\newblock Effective travel time estimation: When historical trajectories over road networks matter.
\newblock In \emph{Proceedings of the 2020 ACM SIGMOD International Conference on Management of Data}, SIGMOD '20, page 2135–2149, New York, NY, USA, 2020. Association for Computing Machinery.
\newblock ISBN 9781450367356.

\bibitem[Zhan et~al.(2013)Zhan, Hasan, Ukkusuri, and Kamga]{ZhanEtAl13}
Zhan, X., Hasan, S., Ukkusuri, S.~V., and Kamga, C.
\newblock Urban link travel time estimation using large-scale taxi data with partial information.
\newblock \emph{Transportation Research Part C: Emerging Technologies}, 33:\penalty0 37--49, 2013.

\bibitem[Zhang et~al.(2021)Zhang, Buehler, Broaddus, and Sweeney]{zhang2021type}
Zhang, W., Buehler, R., Broaddus, A., and Sweeney, T.
\newblock What type of infrastructures do e-scooter riders prefer? a route choice model.
\newblock \emph{Transportation research part D: transport and environment}, 94:\penalty0 102761, 2021.

\bibitem[Zhou et~al.(2023)Zhou, Xiao, Gong, Chen, Fang, Tan, Ma, Li, Hua, Jeon, and Zhang]{ZhouEtAl23}
Zhou, W., Xiao, X., Gong, Y.-J., Chen, J., Fang, J., Tan, N., Ma, N., Li, Q., Hua, C., Jeon, S.-W., and Zhang, J.
\newblock Travel time distribution estimation by learning representations over temporal attributed graphs.
\newblock \emph{IEEE Transactions on Intelligent Transportation Systems}, 24\penalty0 (5):\penalty0 5069--5081, 2023.

\bibitem[Zhu et~al.(1997)Zhu, Byrd, Lu, and Nocedal]{lbfgsb}
Zhu, C., Byrd, R., Lu, P., and Nocedal, J.
\newblock {Algorithm 778: L-BFGS-B: Fortran Subroutines for Large-Scale Bound-Constrained Optimization}.
\newblock \emph{ACM Transactions on Mathematical Software}, 23\penalty0 (4):\penalty0 550–560, 1997.

\bibitem[Zimmermann et~al.(2018)Zimmermann, {Blom V\"astberg}, Frejinger, and Karlstr\"om]{ZimmEtAl18}
Zimmermann, M., {Blom V\"astberg}, O., Frejinger, E., and Karlstr\"om, A.
\newblock Capturing correlation with a mixed recursive logit model for activity-travel scheduling.
\newblock \emph{Transportation Research Part C: Emerging Technologies}, 93:\penalty0 273--291, 2018.

\end{thebibliography}

\clearpage
\appendix

\section{Data Processing}\label{appendix:data}

For the real data, just like \citet{bertsimas2019travel}, we first downloaded a recent map of the city from Open Street Map (OSM) \citep{OpenStreetMap} and cut out the parts outside a polygon that encircles the target (Manhattan). We fill the missing speed limit using official maximum speed limits based on the type of road. We remove all \rev{non-}intersection nodes as long as adding arcs to bypass that node will not create arcs larger than 100 meters. A non-intersection node is a node that has two neighbors such that they are all part of a one-way or two-way path. We keep stop signs and red light nodes unless the crossing is at most 50 meters or 100 nodes away from a red light or stop sign. We split arcs longer than 200 meters into two arcs connected to a new node in the middle of the two nodes. Note that we keep dead ends. We repeat this process until we cannot remove any more nodes. We then take the largest connected component of the graph. This transformation keeps the topology intact as we can easily interpolate the original graph's travel time. This process removes many nodes whose sole purpose is to convey the geometry of the streets, as OSM does not support curved arcs. 

We match the origin and destination from the cab dataset using the nearest node with an upper bound of 100 meters on the distance between the observation and the node. We remove trips that are shorter than 30 seconds or longer than 3 hours or have an average speed of less than 1 meter per second or greater than the city's maximum speed on the highway (50 miles per hour in the case of NYC) on the shortest path connecting the origin and destination. We then take the data observed in a time range of day groups (workdays, weekends, etc.). We show the results in Figure~\ref{fig:nyc_map}. Note that the size of the dataset is roughly 2 gigabytes, but the compressed version with the required columns only takes approximates 150 megabytes. Since we consider a roughly 60-hour window, the data takes even less space.

We split the dataset into training and  validation sets with 100, 1,000, 10,000, 100,000, and 200,000 observations. The test set contains 100,000 observations. We use the data from April 2016 to validate the code and show the results on May 2016.

\begin{figure}[htbp]
    \centering
    \includegraphics[width=0.5\linewidth]{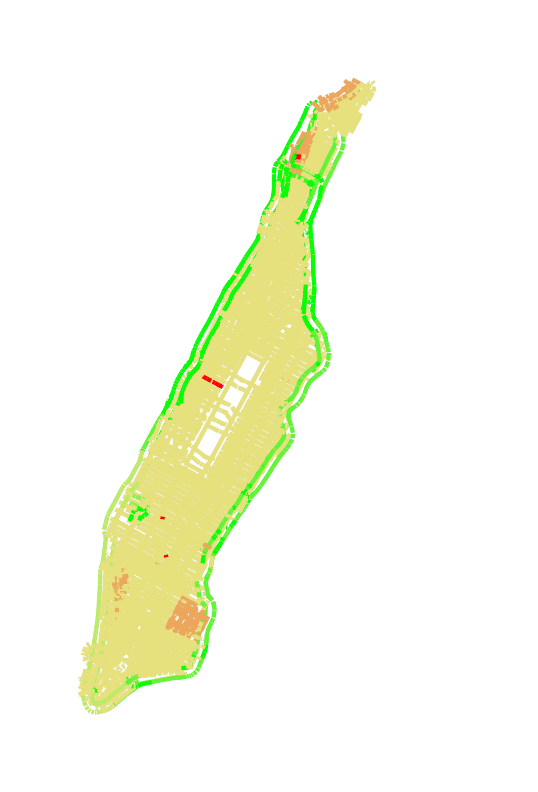}
    \caption{Map of Manhattan at free-flow speed.}
    \label{fig:nyc_map}
\end{figure}

For the synthetic tests, we simulated data on a ten-by-ten grid where each node is 600 meters away from its neighbors. The minimum and maximum speeds are 5.5 and 10 m/s, respectively. The travel time of the arc that ends in the $i$th column and $j$th row is $600/10 - 3 (j - 3300)^2 / 3300^2$ clamped between the legal speed. We use $\pmbeta=(-2, -2, -5)^\top$ to generate 10,000 samples for the train, test, and validation sets. We generate five trajectories between each (o, d) uniformly. We multiply each observed travel time by a sample from the log-normal $(0.1, \sqrt{0.1})$. We show the grid in Figure~\ref{fig:my_labelsd:ground}.

\section{Hyperparameter and Intialization Figures}\label{appendix:fig}
\newcommand{\figratio}[0]{0.85}
In this section, we analyze the effect of hyperparameters on our model. We take the best hyperparameter found on the Manhattan late-morning dataset with 10,000 observations and change one of the hyperparameters in each figure and train the model. We show both the training curve for different parameters and the final validation loss. 

Figures~\ref{fig:bsense} and~\ref{fig:Tsens} show a relatively low sensitivity to the initialization of $\pmbeta$ and $\arctt$, respectively. Figure~\ref{fig:lsense} shows that too much regularization is detrimental and that regularization is probably unnecessary when there is sufficient amount of data. As expected, Figure~\ref{fig:etasense} shows that choosing a too high learning rate can cause suboptimal solutions but choosing a too low learning rate leads to slow convergences.

\begin{figure}[htbp]
    \centering
    \begin{subfigure}{\figratio\linewidth}
    \includegraphics[width=\linewidth]{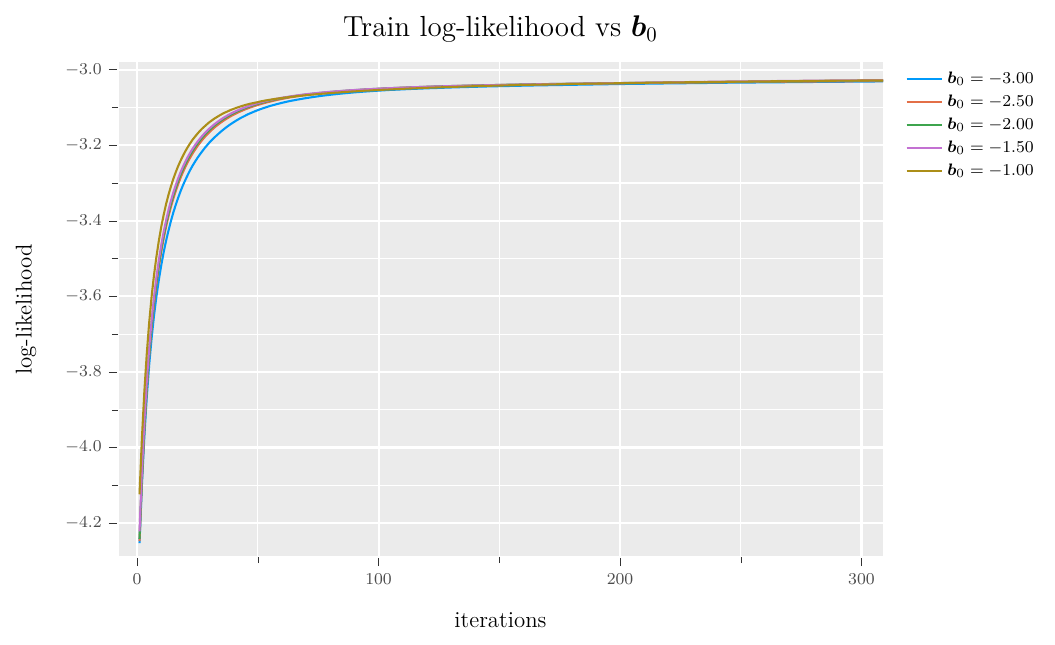}
    \end{subfigure}
    \begin{subfigure}{\figratio\linewidth}
    \includegraphics[width=\linewidth]{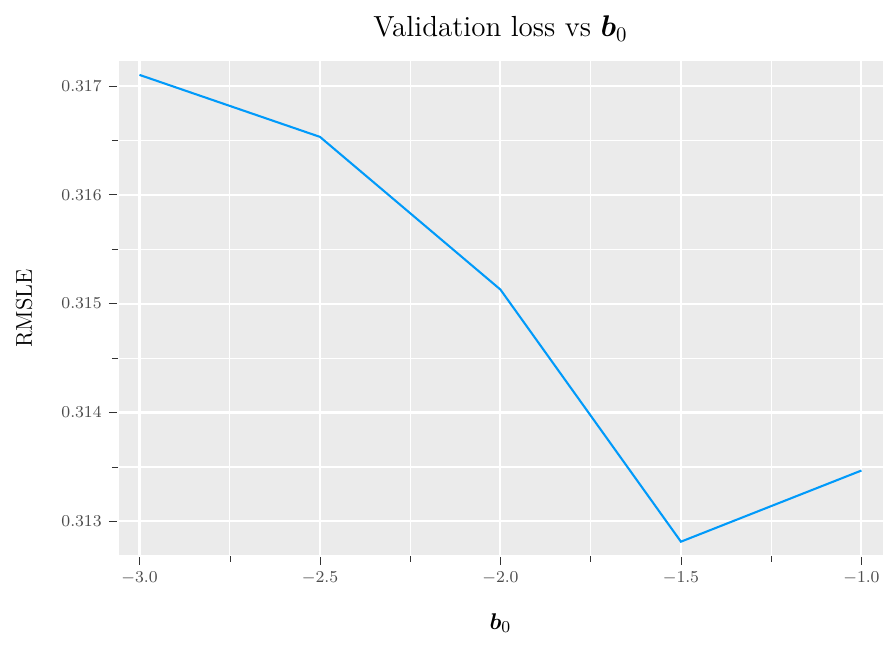}
    \end{subfigure}
    \caption{Sensitivity of our method to the choice of $\pmbeta_0$.}
    \label{fig:bsense}
\end{figure}

\begin{figure}[htbp]
    \centering
    \begin{subfigure}{\figratio\linewidth}
    \includegraphics[width=\linewidth]{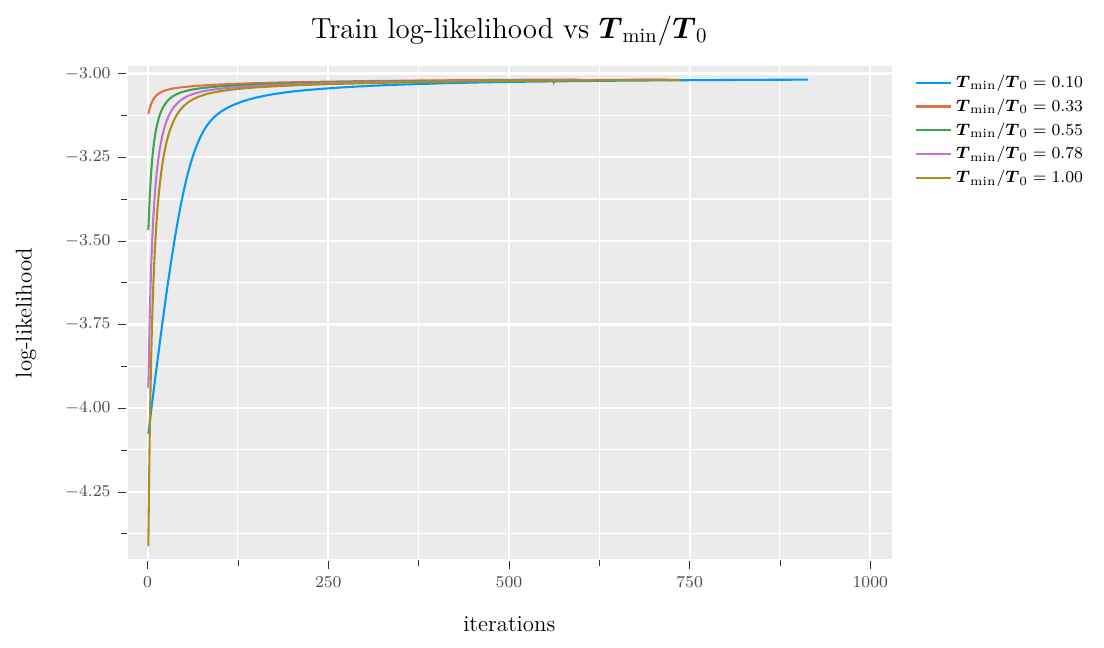}
    \end{subfigure}
    \begin{subfigure}{\figratio\linewidth}
    \includegraphics[width=\linewidth]{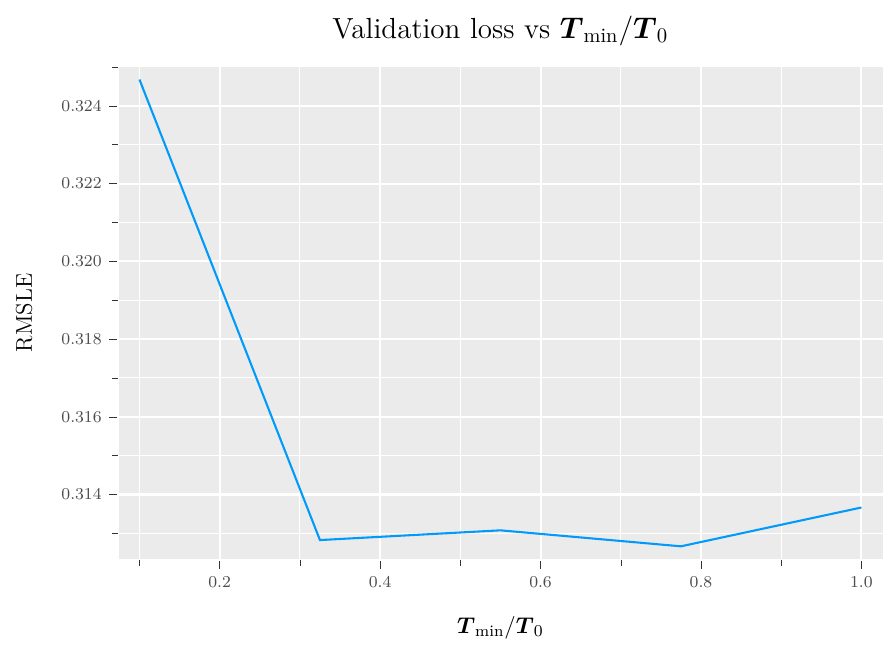}
    \end{subfigure}
    \caption{Sensitivity of our method to the choice of $\arctt_0$.}
    \label{fig:Tsens}
\end{figure}

\begin{figure}[htbp]
    \centering
    \begin{subfigure}{\figratio\linewidth}
    \includegraphics[width=\linewidth]{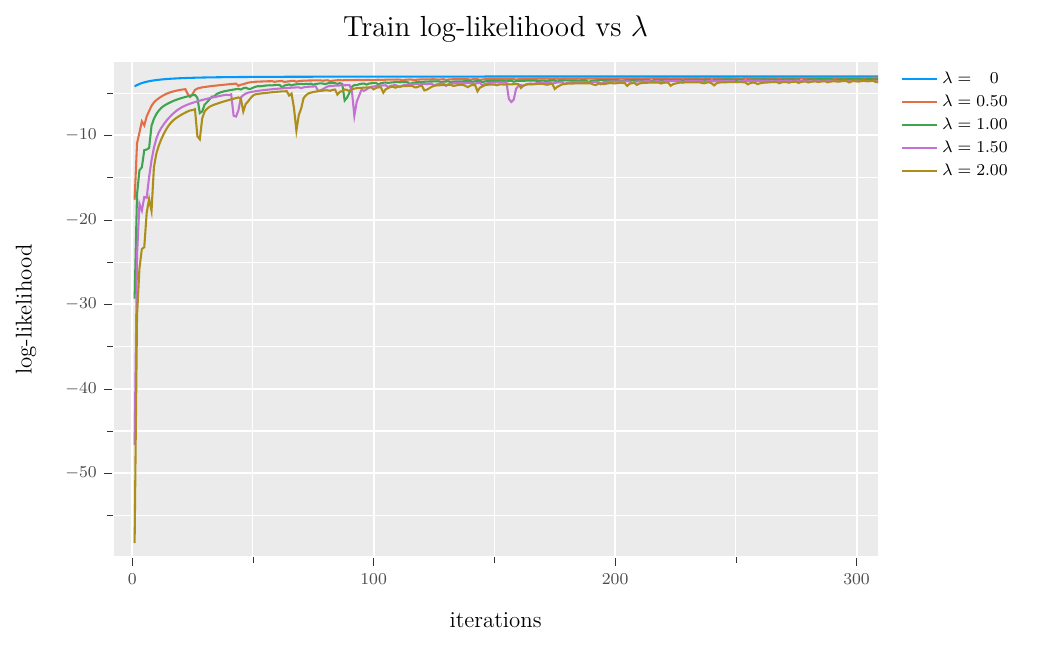}
    \end{subfigure}
    \begin{subfigure}{\figratio\linewidth}
    \includegraphics[width=\linewidth]{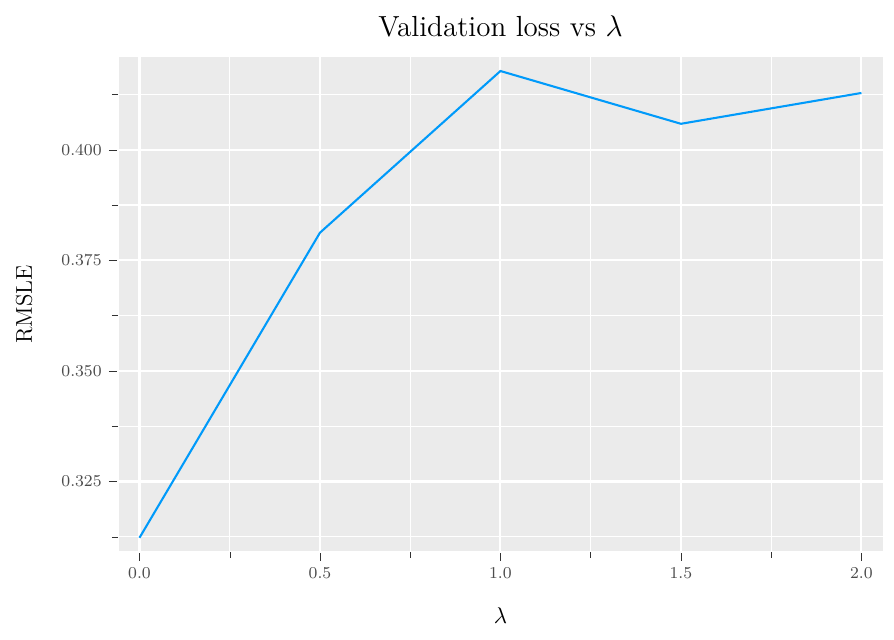}
    \end{subfigure}
    \caption{Sensitivity of our method to the choice of $\lambda$.}
    \label{fig:lsense}
\end{figure}

\begin{figure}[htbp]
    \centering
    \begin{subfigure}{\figratio\linewidth}
    \includegraphics[width=\linewidth]{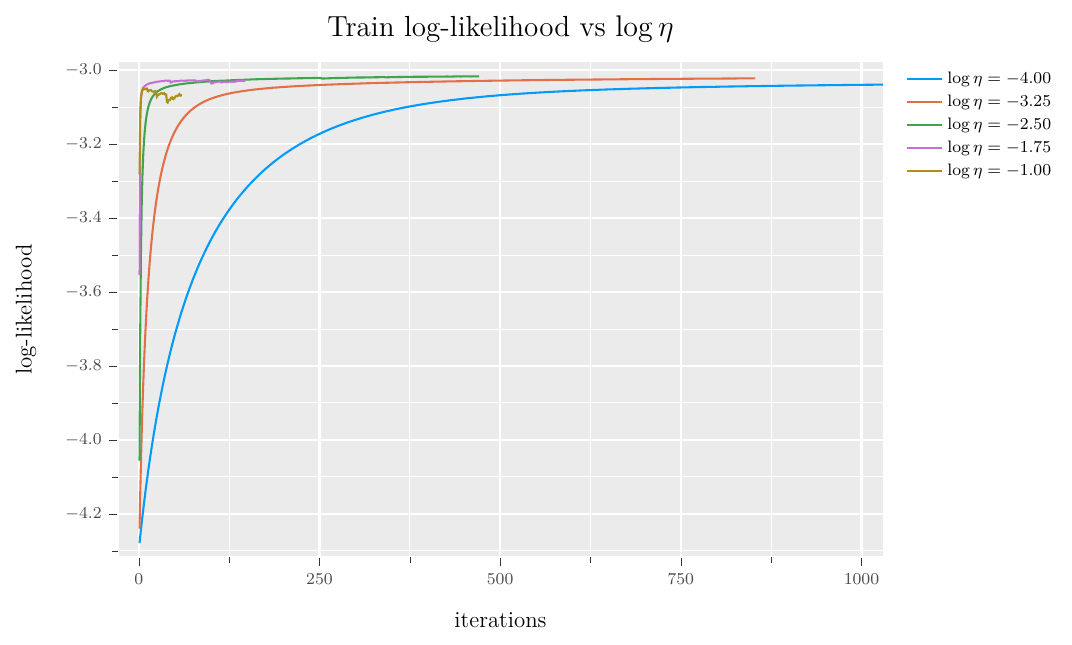}
    \end{subfigure}
    \begin{subfigure}{\figratio\linewidth}
    \includegraphics[width=\linewidth]{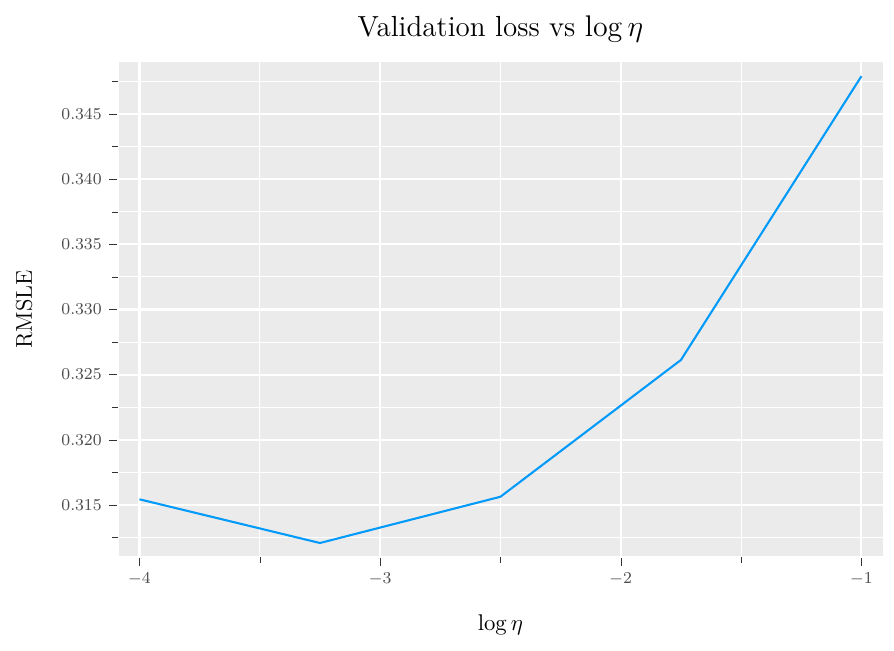}
    \end{subfigure}
    \caption[Sensitivity of our method to the choice of $\eta$.]{Sensitivity of our method to the choice of $\eta$. Logarithmic scale in the learning rate.}
    \label{fig:etasense}
\end{figure}

\section{Parameter Estimates}
Table~\ref{tab:param} contains the estimated path choice models in each settings.

\begin{table}[htbp]
\caption[$\pmbeta$ estimated on NYC.]{$\pmbeta$ estimated on NYC.}
\centering
\begin{tabular}{rccccc}
    \toprule 
size & 
$\hat \pmbeta_\text{non-residential}$ & 
$\hat \pmbeta_\text{residential}$ &
$\hat \pmbeta_\text{intersection}$ &
$\hat \pmbeta_\text{left turn}$ &
$\pmbeta_\text{U-turn}$ (fixed) 
\\\midrule
\multicolumn{6}{c}{Early morning} \\
\midrule
\multirow{3}{*}{100}
 & -2.82  & -2.64  & -0.08  & -0.97  & -5.00 \\
 & -1.39  & -1.18  & -0.39  & -0.49  & -5.00 \\
 & -2.36  & -2.75  & -0.00  & -1.46  & -5.00 \\
\midrule
\multirow{3}{*}{1,000}
 & -2.85  & -2.47  & -0.47  & -1.31  & -5.00 \\
 & -2.72  & -2.68  & -0.15  & -0.86  & -5.00 \\
 & -2.34  & -2.27  & -0.42  & -0.88  & -5.00 \\
\midrule
\multirow{3}{*}{10,000}
 & -2.77  & -3.09  & -0.39  & -0.48  & -5.00 \\
 & -2.79  & -3.21  & -0.38  & -0.55  & -5.00 \\
 & -2.44  & -2.73  & -0.42  & -0.75  & -5.00 \\
\midrule
\multirow{3}{*}{100,000}
 & -3.30  & -3.57  & -0.27  & -0.48  & -5.00 \\
 & -3.91  & -4.63  & -0.30  & -0.54  & -5.00 \\
 & -3.28  & -3.55  & -0.31  & -0.72  & -5.00 \\
\midrule
\multirow{3}{*}{200,000}
 & -4.37  & -4.67  & -0.25  & -0.77  & -5.00 \\
 & -3.41  & -3.63  & -0.32  & -0.65  & -5.00 \\
 & -3.39  & -3.64  & -0.29  & -0.56  & -5.00 \\
\midrule
\multicolumn{6}{c}{Late morning} \\
\midrule
\multirow{3}{*}{100}
 & -1.50  & -1.51  & -0.45  & -0.25  & -5.00 \\
 & -1.14  & -2.10  & -0.46  & -0.17  & -5.00 \\
 & -2.49  & -2.37  & -2.50  & -2.80  & -5.00 \\
\midrule
\multirow{3}{*}{1,000}
 & -1.71  & -1.96  & -0.45  & -0.02  & -5.00 \\
 & -0.97  & -1.11  & -1.00  & -0.44  & -5.00 \\
 & -1.32  & -1.14  & -2.01  & -0.61  & -5.00 \\
\midrule
\multirow{3}{*}{10,000}
 & -2.51  & -2.34  & -0.91  & -0.90  & -5.00 \\
 & -2.28  & -2.26  & -0.88  & -1.10  & -5.00 \\
 & -2.10  & -2.31  & -0.75  & -0.59  & -5.00 \\
\midrule
\multirow{3}{*}{100,000}
 & -2.99  & -3.83  & -0.39  & -0.40  & -5.00 \\
 & -3.20  & -4.04  & -0.32  & -0.44  & -5.00 \\
 & -3.09  & -3.91  & -0.39  & -0.53  & -5.00 \\
\midrule
\multirow{3}{*}{200,000}
 & -3.34  & -4.13  & -0.36  & -0.22  & -5.00 \\
 & -3.73  & -4.78  & -0.46  & -0.54  & -5.00 \\
 & -3.53  & -4.35  & -0.36  & -0.02  & -5.00 \\
\midrule
\end{tabular}
\caption{Route choice model parameters $\pmbeta$ estimated for each of the settings. As we rerun the optimization three times, there are three entries per dataset. \label{tab:param}}
\end{table}
\end{document}